%% file: main.tex
\documentclass[11pt]{article}
\usepackage{fullpage}
\usepackage[margin=1in]{geometry}
\usepackage{xfrac}
\usepackage{caption}
\usepackage{flushend}
\usepackage{url}
\usepackage{amsmath,amsfonts,amsthm}
\usepackage{amssymb}
\usepackage{multicol}
\usepackage{enumitem}
\usepackage{tikz}
\usetikzlibrary{tikzmark}
\usepackage{mathtools}

\AtBeginDocument{%
  \providecommand\BibTeX{{%
    \normalfont B\kern-0.5em{\scshape i\kern-0.25em b}\kern-0.8em\TeX}}}

\usepackage{algorithm}
\usepackage{xspace}
\usepackage[noend]{algpseudocode}

\makeatletter
\newcommand{\algmargin}{\the\ALG@thistlm}
\makeatother
\newlength{\whilewidth}
\algdef{SE}[parWHILE]{parWhile}{EndparWhile}[1]
{\parbox[t]{\dimexpr\linewidth-\algmargin}{%
		\hangindent\whilewidth\strut\algorithmicwhile\ #1\ \algorithmicdo\strut}}{\algorithmicend\ \algorithmicwhile}%
\algdef{SE}[SUBALG]{Indent}{EndIndent}{}{\algorithmicend\ }%
\algtext*{Indent}
\algtext*{EndIndent}
\algnewcommand{\parState}[1]{\State%
	\parbox[t]{\dimexpr\linewidth-\algmargin}{\strut #1\strut}}
\makeatother

\newtheorem{theorem}{Theorem}

\newtheorem{corollary}[theorem]{Corollary}
\newtheorem{lemma}[theorem]{Lemma}

\newtheorem{definition}[theorem]{Definition}
\newtheorem{observation}[theorem]{Observation}

\newcommand{\OPR}{\texttt{OPR}\xspace}
\newcommand{\DTPR}{\texttt{DTPR}\xspace}
\newcommand{\OPRmin}{\texttt{OPR-min}\xspace}
\newcommand{\OPRmax}{\texttt{OPR-max}\xspace}
\newcommand{\DTPRmin}{\texttt{DTPR-min}\xspace}
\newcommand{\DTPRmax}{\texttt{DTPR-max}\xspace}

\begin{document}

\title{The Online Pause and Resume Problem: Optimal Algorithms and An Application to Carbon-Aware Load Shifting}

\author{
  Adam Lechowicz\thanks{University of Massachusetts Amherst.  Email: \texttt{alechowicz@cs.umass.edu}}
  \and
  Nicolas Christianson\thanks{California Institute of Technology.  Email: \texttt{nchristianson@caltech.edu}}
  \and
  Jinhang Zuo\thanks{UMass Amherst \& Caltech.  Email: \texttt{jhzuo@cs.umass.edu}}
  \and
  Noman Bashir\thanks{University of Massachusetts Amherst.  Email: \texttt{nbashir@cs.umass.edu}}
  \and
  Mohammad Hajiesmaili\thanks{University of Massachusetts Amherst.  Email: \texttt{hajiesmaili@cs.umass.edu}}
  \and
  Adam Wierman\thanks{California Institute of Technology.  Email: \texttt{adamw@caltech.edu}}
  \and
  Prashant Shenoy\thanks{University of Massachusetts Amherst.  Email: \texttt{shenoy@cs.umass.edu}}
}
\maketitle

\begin{abstract}
\input{abstract}

\end{abstract}

\newpage

\section{Introduction}
\label{sec:intro}
\input{intro}

\section{Problem Formulation and Preliminaries} %
\label{sec:prob}

\input{problem}

\section{Double Threshold Pause and Resume (DTPR) Algorithm}
\label{sec:dtpr}

\input{dtpr}

\section{Main Results}
\label{sec:summary}

\input{results}

\section{Proofs}
\label{sec:analysis}

\input{analysis}

\section{Case Study: Carbon-Aware Temporal Workload Shifting} %
\label{sec:eval}
\input{experiments}

\section{Related Work} 
\label{sec:rel}
\input{relwork.tex}

\section{Conclusion}
\label{sec:conclusion}
\input{conclusion}

\newpage
\bibliographystyle{alpha}
\bibliography{main}

\newpage

\appendix
\input{appendix}

\end{document}

%% file: abstract.tex
We introduce and study the online pause and resume problem.  In this problem, a player attempts to find the $k$ lowest (alternatively, highest) prices in a sequence of fixed length $T$, which is revealed sequentially.  At each time step, the player is presented with a price and decides whether to accept or reject it.  The player incurs a \textit{switching cost} whenever their decision changes in consecutive time steps, i.e., whenever they pause or resume purchasing.  This online problem is motivated by the goal of carbon-aware load shifting, where a workload may be paused during periods of high carbon intensity and resumed during periods of low carbon intensity and incurs a cost when saving or restoring its state.  It has strong connections to existing problems studied in the literature on online optimization, though it introduces unique technical challenges that prevent the direct application of existing algorithms.  Extending prior work on threshold-based algorithms, we introduce \textit{double-threshold} algorithms for both the minimization and maximization variants of this problem.  We further show that the competitive ratios achieved by these algorithms are the best achievable by any deterministic online algorithm.  Finally, we empirically validate our proposed algorithm through case studies on the application of carbon-aware load shifting using real carbon trace data and existing baseline algorithms.

%% file: intro.tex
This paper introduces and studies the \textit{online pause and resume problem} (\OPR), considering both minimization (\OPRmin) and maximization (\OPRmax) variants. In \OPRmin, a player is presented with time-varying prices in a sequential manner and decides whether or not to purchase one unit of an item at the current price. The player must purchase $k$ units of the item over a time horizon of $T$ and they incur a \textit{switching cost} whenever their decision changes in consecutive time steps, i.e., whenever they pause or resume purchasing. The goal of the player is to minimize their total cost, which consists of the aggregate price of purchasing $k$ units and the aggregate switching cost incurred over $T$ slots. In \OPRmax, the setting is exactly the same, but the goal of the player is to maximize their total profit, and any switching cost they incur is subtracted. In both cases, the price values are revealed to the player one by one in an online manner, and the player has to make a decision without knowing the future values. 

Our primary motivation for introducing \OPR is the emerging importance of carbon-aware computing and, more specifically, carbon-aware temporal workload shifting, which has seen significant attention in recent years \cite{radovanovic2022carbon,acun2022holistic,bashir2021enabling,Wiesner:21}. 
In carbon-aware temporal workload shifting, an interruptible and deferrable workload may be paused during periods of high carbon intensity and resumed during periods of low carbon intensity. The workload needs to be running for $k$ units of time to complete and must be finished before its deadline $T$. However, pausing and resuming the workload typically comes with overheads such as storing the state in memory and checkpointing; hence frequent pausing and resuming is undesirable.  The objective of temporal workload shifting is to minimize the total carbon footprint of running the workload, which includes both the original compute demand and the overhead due to pausing and resuming (a.k.a., the switching cost). The carbon intensity of the electric grid is time-varying due to the intermittency of renewable energy, and thus finding the best pause and resume strategy is challenging due to the unknown future fluctuations of carbon intensity.
Note that \OPR can also capture other potentially interesting applications where pricing changes over time and switching frequently is undesirable. One example is renting spot virtual machines from a cloud service provider in the setting where pricing is set according to supply-demand dynamics~\cite{zhang2017optimal,ambati2020hedge,shastri2016transient}.

On the theory front, the \OPR problem has strong connections to various existing problems in the literature on online optimization. We extensively review the prior literature in Section~\ref{sec:rel} and focus on the most relevant theoretical problems below. 
The \OPR problem is strongly connected to the $k$-search problem~\cite{Lorenz:08,Lee:22}, which belongs to the broader class of online conversion problems~\cite{Sun:21}, a.k.a, time series search and one-way trading~\cite{ElYaniv:01}. In the minimization variant of the $k$-search problem, an online decision-maker aims to buy $k$ units of an item for the least cost over a sequence of time-varying cost values. At each step, a cost value is observed, and the decision is whether or not to buy one unit at the current observed cost without knowing the future values.
In contrast to $k$-search, the \OPR problem introduces the additional component of managing the switching cost, which poses a significant additional challenge in algorithm design.

The existence of the switching cost in \OPR connects it to the well-studied problem of smoothed online convex optimization (\texttt{SOCO})~\cite{Lin:12}, also known as convex function chasing (\texttt{CFC})~\cite{FriedmanLinial:93}, and its generalizations including metrical task systems (\texttt{MTS})~\cite{Borodin:92}. In \texttt{SOCO}, a learner is faced with a sequence of cost functions $f_t$ that are revealed online, and must choose an action $x_t$ after observing $f_t$. Based on that decision, the learner incurs a hitting cost, $f_t(x_t)$ as well as a switching cost, $\|x_t - x_{t-1}\|$, which captures the cost associated with changing the decision between rounds. In contrast to \texttt{SOCO}, \OPR includes the long-term constraint of satisfying the demand of $k$ units over the horizon $T$, which poses a significant challenge not present in \texttt{SOCO}-like problems.

The coexistence of these differentiating factors, namely the \textit{switching cost} and the \textit{long-term deadline constraint}, make \OPR uniquely challenging, and means that prior algorithms and analyses for related problems such as $k$-search and \texttt{SOCO} cannot be directly adapted.

\smallskip
\textbf{Contributions.} 
We introduce online algorithms for the minimization and maximization variants of \OPR and show that our algorithms achieve the best possible competitive ratios. We also evaluate the empirical performance of the proposed algorithms on a case study of carbon-aware load shifting. The details of our contributions are outlined below.  

\paragraph{Algorithmic idea: Double-threshold}
To tackle \OPR, we focus our efforts on online threshold-based algorithms (\texttt{OTA}), the prominent design paradigm for classic problems such as $k$-search~\cite{Lorenz:08,Lee:22}, one-way trading~\cite{ElYaniv:01,Sun:21}, and online knapsack problems~\cite{Zhou:08,sun2022online,Yang2021Competitive}. In the $k$-min search problem, for example, a threshold-based algorithm specifies $k$ threshold values and chooses to trade the $i$-th item only if the current price is less than or equal to the value suggested by the $i$-th threshold value.

Direct application of prior \texttt{OTA} algorithms to \OPR results in undesirable behavior (such as frequently changing decisions) since their threshold function design is oblivious to the switching cost present in \OPR. To address this challenge, we seek an algorithm that can simultaneously achieve the following behaviors: (1) when the player is in ``trading mode,'' they should not impulsively switch away from trading in response to a price that is only slightly worse, since this will result in a switching penalty; and (2) the player should not switch to ``trading mode'' unless prices are sufficiently good to warrant the switching cost.
These two ideas motivate an algorithm design that uses two distinct threshold functions, each of which captures one of the above two cases. We present our algorithms \DTPRmin and \DTPRmax for \OPRmin and \OPRmax, respectively, in Section~\ref{sec:dtpr}, which build upon this high-level idea of a double-threshold.

\paragraph{Main results}
While \texttt{OTA} algorithms are intuitive and simple to describe, it is highly challenging to design threshold functions that lead the corresponding algorithms to be competitive against the offline optimum. The addition of switching cost in \OPR further exacerbates the technical challenge of designing optimal threshold functions. 
The key result which enables our double-threshold approach is a technical observation (see Observation~\ref{obs:thresholdDifference}), which shows that the difference between the functions guiding the algorithm's decisions should be a factor of~$\beta$, where $\beta$ represents the fixed switching cost incurred by changing the decision in \OPR.

Identifying this relationship between the two threshold functions significantly facilitates the competitive analysis of both \DTPRmin and \DTPRmax, enabling our derivation of a closed form of each threshold. Using this idea, we characterize the competitive ratios of \DTPRmin and \DTPRmax as a function of problem parameters, including an explicit dependence on the magnitude of the switching cost $\beta$ (see Theorems~\ref{thm:compPRmin} and~\ref{thm:compPRmax}). 
Furthermore, we derive lower bounds for the competitive ratio of any deterministic online algorithm, showing that our proposed algorithms are optimal for this problem (formal statements in Theorems~\ref{thm:lowerboundmin} and~\ref{thm:lowerboundmax}). The competitive ratios we derive for both \DTPRmin and \DTPRmax exactly recover the best prior competitive results for the $k$-search problem~\cite{Lorenz:08}, which corresponds to the case of $\beta=0$ in \OPR, i.e., no switching cost. Formal statements and a more detailed discussion of our main results are presented in Section~\ref{sec:summary}.

\paragraph{Case study.}
Finally, in Section~\ref{sec:eval}, we illustrate the performance of our proposed algorithm by conducting an experimental case study simulating the carbon-aware load shifting problem. We utilize real-world \textit{carbon traces} from Electricity Maps~\cite{electricity-map}, which contain carbon intensity values for grid-sourced electricity across the world.  Our experiments simulate different strategies for scheduling a deferrable and interruptible workload in the face of uncertain future carbon intensity values.  We show that our algorithm's performance significantly improves upon existing baseline methods and adapted forms of algorithms for related problems such as $k$-min search.  

\smallskip

%% file: problem.tex
We begin by formally introducing the \OPR problem and providing background on the online threshold-based algorithm design paradigm, which is used in the design of our proposed algorithms. Table~\ref{tab:notations} summarizes the core notations for \OPR.  Recall that this formulation is motivated by the setting of carbon-aware temporal workload shifting, as described in the introduction.

\subsection{Problem Formulation}
There are two variants of the online pause and resume problem (\OPR).\footnote{We use \OPR whenever the context is applicable to both minimization (\OPRmin) and maximization (\OPRmax) variants of the problem, otherwise, we refer to the specific variant. The same policy applies to \DTPR, our proposed algorithm for \OPR.} In \OPRmin (\OPRmax) a player must buy (sell) $k \geq 1$ units of some asset (one unit at each time step) with the goal of minimizing (maximizing) their total cost (profit) within a time horizon of length $T$.  At each time step $1 \leq t \leq T$, the player is presented with a price $c_t$, and must immediately decide whether to accept this price ($x_t = 1$) or reject it ($x_t = 0$).  The player is required to complete this transaction for all $k$ units by some point in time $T$. Both $k$ and $T$ are known in advance. 
Thus, the requirement of $k$ transactions is a hard constraint, i.e., $\sum_{t=1}^T x_t = k$, and if at time $T - i$ the player still has $i$ units remaining to buy/sell, they must accept the prices in the subsequent $i$ slots to accomplish $k$ transactions. 

Additionally, in both variants of \OPR, the player incurs a \emph{fixed switching cost} $\beta > 0$ whenever they decide to change decisions between two adjacent time steps (i.e., when $\lVert x_{t-1} - x_t \rVert = 1$).  We assume that $x_0 = 0$ and $x_{T+1} = 0$, implying that any player must incur a minimum switching cost of $2\beta$, once for switching ``on'' and once for switching ``off''.  While the player incurs at least a switching cost of $2\beta$, note that the total switching cost incurred by the player is bounded by the size of the asset $k$ since the switching cost cannot be larger than $ k2\beta$. %

\noindent In summary, the offline version of \OPRmin can be summarized as follows:
\begin{align}
\min & \underbrace{ \left( \sum_{t=1}^T c_tx_t \right)}_{\text{Accepted prices}}  + \underbrace{ \left( \sum_{t=0}^{T+1}\beta||x_t - x_{t-1}|| \right)}_{\text{Switching cost}}, \quad\text{s.t., }  \underbrace{\sum_{t=1}^T x_t = k,}_{\text{Deadline constraint}} \ x_t \in \{0, 1\}, \ \forall t \in [1, T], \label{align:objMin}
\end{align}
while the offline version of \OPRmax is
\begin{align}
\max & \left( \sum_{t=1}^T c_tx_t \right) - \left( \sum_{t=0}^{T+1}\beta||x_t - x_{t-1}|| \right) , \quad\text{s.t., } \quad \sum_{t=1}^T x_t = k, \quad x_t \in \{0, 1\}, \ \forall t \in [1, T]. \label{align:objMax}
\end{align}

\begin{table}[t]
	\caption{A summary of key notations }
 \vspace{-3mm}
	\label{tab:notations}
	\begin{center}
		\begin{tabular}[P]{|c|l|}
            \hline
            \textbf{Notation} & \textbf{Description} \\
            \hline
            $k \in \mathbb{N}$ & Number of units which must be bought (or sold) \\
			\hline
            $T$ & Deadline constraint; the player must buy (or sell) $k$ units before time $T$ \\
            \hline
            $t \in [1, T]$ & Current time step \\
            \hline
            $x_t \in \{0, 1\}$ & Decision at time $t$. $x_t = 1$ if price $c_t$ is accepted, $x_t = 0$ if $c_t$ is not accepted \\
			\hline
            $\beta$ & Switching cost incurred when algorithm's decision $x_t \not = x_{t-1}$ \\
			\hline
            $U$ & Upper bound on any price that will be encountered  \\
            \hline
            $L$ & Lower bound on any price that will be encountered \\
            \hline
            $\theta = U/L$ & Price fluctuation ratio \\
			\hline
            \hline
			$c_t$ & (\textit{Online input}) Price revealed to the player at time $t$ \\
			\hline
			$c_{\min} \; \& \; c_{\max}$ & (\textit{Online input}) The actual minimum and maximum prices in a sequence  \\
			\hline
		\end{tabular}
	\end{center}
\end{table}

Of course, our focus is the online version of \OPR, where the player must make irrevocable decisions at each time step without the knowledge of future inputs. More specifically, in both variants of \OPR the sequence of prices $\{c_t\}_{t \in [1, T]}$ is revealed sequentially -- future prices are \textit{unknown} to an online algorithm, and each decision $x_t$ is irrevocable.

\paragraph{Competitive analysis}
Our goal is to design an online algorithm that maintains a small \textit{competitive ratio}~\cite{Borodin:92}, i.e., performs nearly as well as the offline optimal solution.   For an online algorithm $\texttt{ALG}$ and an offline optimal solution $\texttt{OPT}$, the competitive ratio for a minimization problem is defined as:
$
\textnormal{CR}(\texttt{ALG}) = \max_{\sigma \in \Omega} \texttt{ALG}(\sigma) / \texttt{OPT}(\sigma),
$
where $\sigma$ denotes a valid input sequence for the problem and $\Omega$ is the set of all feasible input instances. Further, $\texttt{OPT}(\sigma)$ is the optimal cost given this input, and $\texttt{ALG}(\sigma)$ is the cost of the solution obtained by running the online algorithm over this input.  
Conversely, for a problem with a maximization objective, the competitive ratio is defined as $\max_{\sigma \in \Omega} \texttt{OPT}(\sigma)/\texttt{ALG}(\sigma)$.  With these definitions,  the competitive ratio for both minimization and maximization problems is always greater than or equal to one, and the lower the better. %

\smallskip

\paragraph{Assumptions and additional notations.} We make no assumptions on the underlying distribution of the prices other than the assumption that the set of prices arriving online $\{ c_t \}_{t \in [1, T]}$ has bounded support, i.e., $c_t \in [L, U] \forall t \in [1,~T]$, where $L$ and $U$ are known to the player.  We also define $\theta = U/L$ as the \textit{price fluctuation}.  These are standard assumptions in the literature for many online problems, including one-way trading, online search, and online knapsack; and without them the competitive ratio of any algorithm is unbounded.  We use  $c_{\min}(\sigma) = \min_{t \in [1, T]} c_t$ and $c_{\max}(\sigma) = \max_{t \in [1, T]} c_t$ to denote the minimum and maximum encountered prices for any valid \texttt{OPR} sequence~$\sigma$.

\subsection{Background: Online Threshold-Based Algorithms (\texttt{OTA})} \label{sec:OTA}

Online threshold-based algorithms (\texttt{OTA}) are a family of algorithms for online optimization in which a carefully designed \textit{threshold function} is used to specify the decisions made at each time step. At a high level, the threshold function defines the ``minimum acceptable quality'' that an arriving input/price must satisfy in order to be accepted by the algorithm. The threshold is chosen specifically so that an agent greedily accepting prices meeting the threshold at each step will be ensured a competitive guarantee.
This algorithmic framework has seen success in the online search and one-way trading problems \cite{Lee:22,Sun:21,Lorenz:08,ElYaniv:01} as well as the related online knapsack problem \cite{Zhou:08,sun2022online,Yang2021Competitive}.  In these works, the derived threshold functions are optimal in the sense that the competitive ratios of the resulting threshold-based algorithms match information-theoretic lower bounds of the corresponding online problems. As discussed in the introduction, the framework does not apply directly to the \OPR setting, but we make use of ideas and techniques from this literature. We briefly detail the most relevant highlights from the prior results before discussing how these related problems generalize to \OPR in the next section. 

\smallskip
\paragraph{1-min/1-max search.}  In the online 1-min/1-max search problem, a player attempts to find the single lowest (respectively, highest) price in a sequence, which is revealed sequentially.  The player's objective is to either minimize their cost or maximize their profit.  When each price arrives, the player must decide immediately whether to accept the price, and the player is forced to accept exactly one price before the end of the sequence.
For this problem, El-Yaniv et al.~\cite{ElYaniv:01} presents a deterministic threshold-based algorithm. The algorithm assumes a finite price interval, i.e., the price is bounded by the interval $[L,U]$, where $L$ and $U$ are known. Then, it sets a constant threshold $\Phi = \sqrt{LU}$,  and the algorithm simply selects the first price that is less than or equal to $\Phi$ (for the maximization version, it accepts the first price greater than or equal to $\Phi$). This algorithm achieves a competitive ratio of $\sqrt{U/L} = \sqrt{\theta}$, which matches the lower bound; hence, it is optimal~\cite{ElYaniv:01}.

\smallskip
\paragraph{$k$-min/$k$-max search.} The online $k$-min/$k$-max search problem extends the 1-min/1-max search problem -- a player attempts to find the $k$ lowest (conversely, highest) prices in a sequence of prices revealed sequentially.  The player's objective is identical to the 1-min/1-max problem, and the player must accept at least $k$ prices by the end of the sequence.  Several works have developed a known optimal deterministic threshold-based algorithm for this problem, including \cite{Lorenz:08, ElYaniv:01}.  Leveraging the same assumption of a finite price interval $[L, U]$, the threshold function is a sequence of $k$ thresholds $\{\Phi_i\}_{i \in [1, k]}$, which is also called the \textit{reservation price policy}.  At each step, the algorithm accepts the first price, which is less than or equal to $\Phi_i$, where $i - 1$ is the number of prices that have been accepted thus far (for the maximization version, it accepts the first price which is $\geq \Phi_i$).
In the $k$-min setting, this algorithm is $\alpha$-competitive, where $\alpha$ is the unique solution of 

\begin{align}
     \smash{\frac{1 - 1/\theta}{1-1/\alpha} = \left( 1 + \frac{1}{\alpha k} \right)^k.} {\normalsize \label{eq:kminalpha}}
\end{align}

For the $k$-max variant, this algorithm is $\omega$-competitive, where $\omega$ is the unique solution of 

\begin{align}
     \smash{\frac{\theta - 1}{\omega - 1} = \left( 1 + \frac{\omega}{k} \right)^k.} {\normalsize \label{eq:kmaxomega}}
\end{align}

The sequence of thresholds $\{\Phi_i\}_{i \in [1, k]}$ for both variants of the problem are constructed by analyzing possible input cases, ``hedging'' against the risk that future (unknown) prices will jump to the worst possible value, i.e., $U$ for $k$-min search, $L$ for $k$-max search.  These potential cases can be enumerated for different values of $i$, where $0 \leq i \leq k$ denotes the number of prices accepted so far.  By simultaneously \textit{balancing} the competitive ratios for each of these cases (setting each ratio equal to the others), the optimal threshold values and the optimal competitive ratios are derived.  We refer to this technique as the \textit{balancing rule} and a rigorous proof of this approach, with corresponding lower bounds, can be found in \cite{Lorenz:08}. The lower bounds highlight that the $\alpha$ and $\omega$ which solve the expressions for the competitive ratios above are optimal for any deterministic $k$-min and $k$-max search algorithms, respectively. Further,  $\alpha$ and $\omega$ provide insight into a fundamental difference between the minimization and maximization settings of $k$-search.  As discussed in \cite{Lorenz:08}, for large $\theta$, the best algorithm for $k$-max search is roughly $O ( k \sqrt[k]{\theta} )$-competitive, while the best algorithm for $k$-min search is at best $O ( \sqrt{\theta} )$-competitive.  Similarly, for fixed $\theta$ and large $k$, the optimal competitive ratio for $k$-max search is roughly $O \left( \ln \theta \right)$, while the optimal competitive ratio for $k$-min search converges to $O ( \sqrt{\theta} )$.  

%% file: dtpr.tex
A fundamental challenge in algorithm design for \OPR is how to characterize threshold functions that incorporate the presence of switching costs in their design. 
Our key algorithmic insight is to incorporate the switching cost into the threshold function by defining \textit{two distinct threshold functions}, where the function to be used for price admittance changes based on the current state (i.e., whether or not the previous price was accepted by the algorithm).

To provide intuition for the state-dependence of the threshold function, consider the setting of \OPRmin. At a high level, if the player has not accepted the previous price, they should wait to accept anything until prices are sufficiently low to justify incurring a cost to switch decisions.  On the other hand, if the player has accepted the previous price, they might be willing to accept a slightly higher price -- if they do not accept this price, they will incur a cost to switch decisions. While this high-level idea is intuitive, characterizing the form of threshold functions such that the resulting algorithms are competitive is challenging.

\begin{algorithm}[t]
\caption{Double Threshold Pause and Resume for \OPRmin (\DTPRmin)}\label{alg:min}
\begin{algorithmic}[1]
\Require threshold values $\{\ell_i\}_{i \in [1, k]}$ and $\{u_i\}_{i \in [1, k]}$ defined in Eq.~\eqref{eq:minthres}, deadline $T$
\Ensure online decisions $\{x_t\}_{t \in [1, T]}$
\State \textbf{initialize: } i = 1;
\While{price $c_t$ arrives \textbf{and} $i \leq k$}
\If{$(k - i) \geq (T - t)$} \Comment{\footnotesize\texttt{close to the deadline $T$, we must accept remaining prices}\normalsize}
    \State price $c_t$ is accepted, set $x_t = 1$ \label{line:min-force}
\ElsIf{$x_{t-1} = 0$} \Comment{\footnotesize\texttt{If previous price was not accepted}\normalsize}
    \If{$c_t \leq \ell_i$} \; price $c_t$ is accepted, set $x_t = 1$
    \Else \; price $c_t$ is rejected, set $x_t = 0$
    \EndIf
\ElsIf{$x_{t-1} = 1$} \Comment{\footnotesize\texttt{If previous price was accepted}\normalsize}
    \If{$c_t \leq u_i$}\; price $c_t$ is accepted, set $x_t = 1$
    \Else\; price $c_t$ is rejected, set $x_t = 0$
    \EndIf
\EndIf
\State \textnormal{update} $i = i + x_t$
\EndWhile
\end{algorithmic}
\end{algorithm}

\paragraph{The \DTPRmin algorithm}
Our proposed algorithm, Double Threshold Pause and Resume (\texttt{DTPR}) for \OPRmin is summarized in Algorithm~\ref{alg:min}. 
Prior to any prices arriving online, \DTPRmin computes two families of threshold values, $\{\ell_i\}_{i \in [1, k]}$ and $\{u_i\}_{i \in [1, k]}$, where $\ell_i \leq u_i \ \forall i \in [1, k]$, whose values are defined as follows.

\begin{definition}[\DTPRmin~Threshold Values] \label{def:minthres}
For each integer $i$ on the interval $[1, k]$, the following expressions give the corresponding threshold values of $u_i$ and $\ell_i$ for \DTPRmin.
\begin{align}
u_{i} = U \left[ 1 - \left(1 - \frac{1}{\alpha} \right) \left( 1 + \frac{1}{k\alpha} \right)^{i-1} \right] + 2\beta \left[ \left(\frac{1}{k \alpha} - \frac{1}{k} + 1 \right) \left( 1 + \frac{1}{k\alpha} \right)^{i-1} \right], \;\;\;\;\; \ell_i = u_i - 2\beta, \label{eq:minthres}
\end{align}
where $\alpha$ is the competitive ratio of \DTPRmin defined in Equation~\eqref{eq:alpha}.
\end{definition}

The role of these thresholds is to incorporate the switching cost into the algorithm's decisions, and to alter the acceptance criteria of \DTPRmin based on the current state.  For \OPRmin, the current state is \textit{whether the previous item was accepted}, i.e., whether $x_{t-1}$ is $0$ or $1$. 
As prices are sequentially revealed to the algorithm at each time $t$, the $i$th price accepted by \DTPRmin will be the first price which is at most $\ell_i$ if $x_{t-1} = 0$, or at most $u_i$ if $x_{t-1} = 1$.
Note that, as indicated in Line~\ref{line:min-force}, \DTPRmin may be forced to accept the last prices of the sequence, which can be ``worse'' than the current threshold values, to satisfy the deadline constraint of \texttt{OPR}.

\paragraph{The \DTPRmax algorithm}
Pseudocode is summarized in the appendix, in Algorithm \ref{alg:max}. The logical flow of \DTPRmax shares a similar structure to that of \DTPRmin, with a few important differences highlighted here. For \OPRmax, the $i$th price accepted by \DTPRmax will be the first price which is at least $u_i$ if $x_{t-1} = 0$, or at least $\ell_i$ if $x_{t-1} = 1$. Further, the threshold functions are defined as follows.

\begin{definition}[\DTPRmax~Threshold Values] \label{def:maxthres}
For each integer $i$ on the interval $[1, k]$, the following expressions give the corresponding threshold values of $\ell_i$ and $u_i$ for \DTPRmax.
\begin{align}
\ell_{i} = L \left[ 1 + \left(\omega - 1 \right) \left( 1 + \frac{\omega}{k} \right)^{i-1} \right] - 2\beta \left[ \left(\frac{\omega}{k} - \frac{1}{k} + 1 \right) \left( 1 + \frac{\omega}{k} \right)^{i-1} \right], \;\;\;\;\; u_i = \ell_i + 2\beta, \label{eq:maxthres}
\end{align}
where $\omega$ is the competitive ratio of \emph{\DTPRmax} defined in Equation~\eqref{eq:omega}.
\end{definition}

\begin{figure*}[t]
	\minipage{0.48\textwidth}
    \centering
	\includegraphics[width=\linewidth]{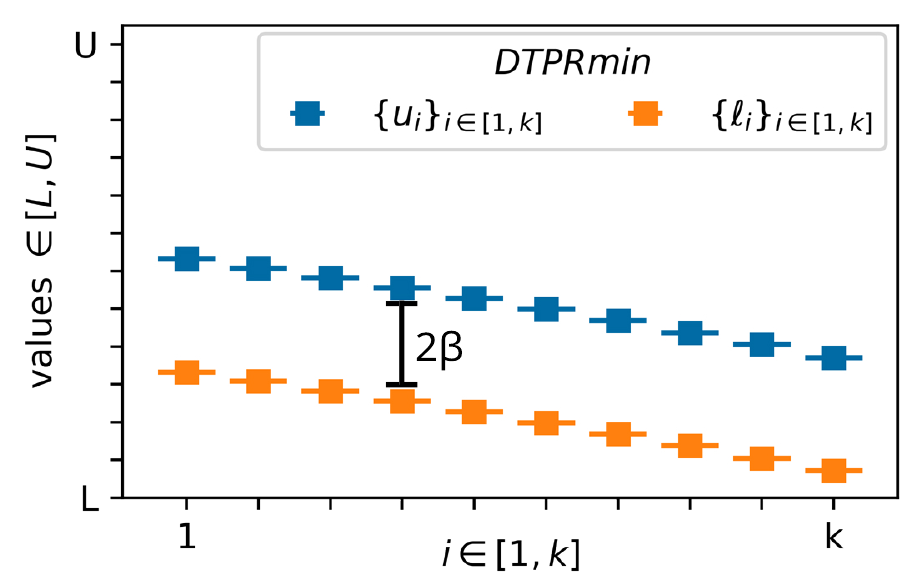}\vspace{-1em} 
    \caption{\DTPRmin thresholds $\ell_i$ and $u_i$ for $i \in~[1,k]$ plotted using example parameters ($k = 10$). } \label{fig:dtprmin}
	\endminipage\hfill
    \minipage{0.48\textwidth}
    \centering
	\includegraphics[width=\linewidth]{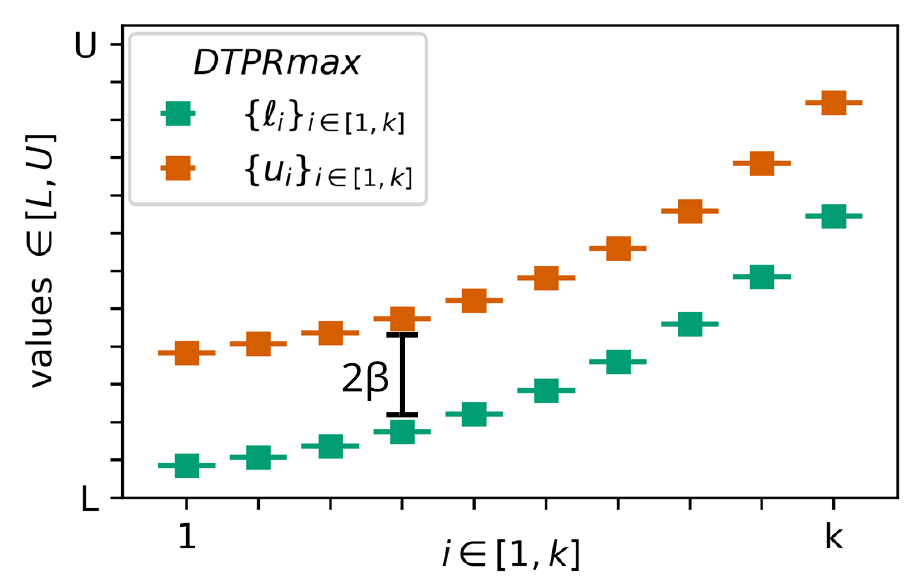}\vspace{-1em}
    \caption{\DTPRmax thresholds $u_i$ and $\ell_i$ for $i \in~[1,k]$ plotted using example parameters ($k = 10$). } \label{fig:dtprmax}
	\endminipage\hfill
\end{figure*}

In Figures \ref{fig:dtprmin} and \ref{fig:dtprmax}, we plot threshold values for \DTPRmin~and \DTPRmax, respectively, using example parameters of $U = 30, L = 5, k = 10,$ and $\beta = 3$.  We annotate the difference of $2\beta$ between $\ell_i$ and $u_i$; recall that each of these thresholds corresponds to a \textit{current state} for \DTPR, i.e. whether the previous item was accepted.  Note that the \DTPRmin~threshold values \textit{decrease} as $k$ gets larger, while the \DTPRmax~threshold values \textit{increase} as $k$ gets larger.  At a high-level, each $i$th threshold ``hedges'' against a scenario where none of the future prices meet the current threshold.  In this case, even if the algorithm is forced to accept the \textit{worst possible prices} at the end of the sequence, we want competitive guarantees against an offline \texttt{OPT}.  Such guarantees rely on the fact that in the worst-case, \texttt{OPT} cannot accept prices that are all significantly better than \DTPR's $i$th ``unseen'' threshold value because such prices did not exist in the sequence.  

\subsection*{Designing the Double Threshold Values}  A key component of the \DTPR algorithms for both variants are the thresholds in Equations~\eqref{eq:minthres} and~\eqref{eq:maxthres}.  The key idea is to design the thresholds by incorporating the switching cost into the balancing rules as a hedge against possible worst-case scenarios. 
To accomplish this, we enumerate three difficult cases that \texttt{DTPR} may encounter. (\texttt{CASE-1}): Consider an input sequence where \texttt{DTPR} does not accept any prices before it is forced to accept the last $k$ prices. Here, the enforced prices in the worst-case sequence will be $U$ for \OPRmin and $L$ for \OPRmax.  This sequence occurs only if no price in the sequence meets the first threshold for acceptance. On the other hand, in the case that \texttt{DTPR} does accept prices before the end of the sequence, we can further divide the possible sequences into two extreme cases for the \textit{switching cost} it incurs.  (\texttt{CASE-2}): In one extreme, the algorithm incurs only the minimum switching cost of $2\beta$, meaning that $k$ contiguous prices are accepted by \texttt{DTPR}.  (\texttt{CASE-3}): In the other extreme, \texttt{DTPR} incurs the maximum switching cost of $k2\beta$, meaning that $k$ non-contiguous prices are accepted. Intuitively, in order for $\texttt{DTPR}$ to be competitive in either of these extreme cases, the prices accepted in the latter case should be sufficiently ``good'' to absorb the extra switching cost of $(k-1)2 \beta$.  

Given the insight from these cases, we use can use the balancing rule (see Section \ref{sec:OTA}) to derive the two threshold families.  Let $\sigma$ be any arbitrary sequence for \OPR. Given these extreme input sequences, we now concretely show how to write the balancing rule equations. We consider the cases of \DTPRmin and \DTPRmax separately below.

\paragraph{Balancing equations for \DTPRmin}

To balance between possible inputs for \OPRmin, consider the following examples for three different values of $c_{\min}(\sigma) > \ell, \ell = \{\ell_1,\ell_2,\ell_3\}$.  If $c_{\min}(\sigma) > \ell_i$, we know that $\texttt{OPT}$ cannot do better than $k\ell_i + 2\beta$. Suppose that $\alpha$ is the target competitive ratio, and we balance between these and other potential cases:
\begin{align}
&\frac{\DTPRmin(\sigma)}{\texttt{OPT}(\sigma)} \leq \underbrace{\frac{kU+2\beta}{k\ell_1 + 2\beta}}_{c_{\min}(\sigma) > \ell_1} = \underbrace{\frac{\ell_1 + (k-1)U + 4\beta}{k \ell_2 + 2\beta} = \frac{u_1 + (k-1)U + 2\beta}{k \ell_2 + 2\beta}}_{c_{\min}(\sigma) > \ell_2} \dots \label{eq:balancemin} \\
&\dots = \underbrace{\frac{\ell_1 +\ell_2 + (k-2)U + 6\beta}{k \ell_3 + 2\beta} = \frac{\ell_1 + u_2 + (k-2)U + 4\beta}{k \ell_3 + 2\beta} = \frac{u_1 + u_2 + (k-2)U + 2\beta}{k \ell_3 + 2\beta}}_{c_{\min}(\sigma) > \ell_3} = \dots = \alpha. \nonumber
\end{align}

As an example, consider $c_{\min}(\sigma) > \ell_2$ and the corresponding cases enumerated above.  Suppose $\DTPRmin$ accepts one price before the end of the sequence $\sigma$, and the other prices accepted are all $U$.  In the first case, where the competitive ratio is $\frac{\ell_1 + (k-1)U + 4\beta}{k \ell_2 + 2\beta}$, we consider the scenario where $\DTPRmin$ switches twice: once to accept the price $\ell_1$, and once to accept $(k-1)$ prices at the end of the sequence, incurring switching cost of $4\beta$.  

In the second case, where the competitive ratio is $\frac{u_1 + (k-1)U + 2\beta}{k \ell_2 + 2\beta}$, we consider the hypothetical scenario where $\DTPRmin$ only switches once to accept some value $u_1$ followed by $(k-1)$ prices at the end of the sequence, incurring switching cost of $2\beta$.  By enumerating cases in this fashion for the other possible values of $c_{\min}(\sigma)$,  we derive a relationship between the lower thresholds $\ell_i$ and the upper thresholds $u_i$ in terms of the switching cost.

\smallskip

\paragraph{Balancing equations for \DTPRmax}
The same idea extends to balance between possible inputs for \OPRmax.  Consider the following examples for a few values of $c_{\max}(\sigma)$.  If $c_{\max}(\sigma) < u_i$, we know that $\texttt{OPT}$ cannot do better than $ku_i - 2\beta$.  Suppose that $\omega$ is the target competitive ratio, and we balance between these and other potential cases:
\begin{align}
&\frac{\texttt{OPT}(\sigma)}{\DTPRmax(\sigma)} \leq \underbrace{\frac{ku_1 - 2\beta}{kL-2\beta}}_{c_{\max}(\sigma) < u_1} = \underbrace{\frac{ku_2 - 2\beta}{u_1 + (k-1)L - 4\beta} = \frac{ku_2 - 2\beta}{\ell_1 + (k-1)L - 2\beta}}_{c_{\max}(\sigma) < u_2} \dots \label{eq:balancemax} \\
&\dots = \underbrace{\frac{ku_3 - 2\beta}{u_1 + u_2 + (k-2)L - 6\beta} = \frac{ku_3 - 2\beta}{u_1 + \ell_2 + (k-2)L - 4\beta} = \frac{ku_3 - 2\beta}{\ell_1 + \ell_2 + (k-2)L - 2\beta}}_{c_{\max}(\sigma) < u_3} = \dots = \omega. \nonumber
\end{align}

\smallskip

\paragraph{Solving for the threshold values}
Given the above balancing equations for both the minimization and maximization variants, the next step is to solve for the unknown values of $\ell_i$ and $u_i$. The following observation summarizes the key insight that enables this.  We show that one can express each $\ell_i$ in terms of $u_i$ and $\beta$, which facilitates the analysis required to solve for thresholds in each balancing equation (given by Equations~\eqref{eq:balancemin} and~\eqref{eq:balancemax}).

\begin{observation}\label{obs:thresholdDifference}
By letting $u_i = \ell_i + 2\beta \;\; \forall i \in [1, k]$, we obtain each possible worst-case permutation of $\ell_i$ thresholds, $u_i$ thresholds, and switching cost. Let $y \in [1, k-1]$ denote the number of switches incurred by \texttt{DTPR}.\\ 
For \DTPRmin, suppose that $c_{\min}(\sigma) > \ell_{j+1}$.  By the definition of \DTPRmin, we know that accepting any $u_i$ helps avoid a switching cost of $+ 2\beta$ in the worst case. Thus,
\begin{align*}
    \sum_{i=0}^j u_i + (k-j) U + 2\beta = \underbrace{\ell_i + \dots}_{y} + \underbrace{u_i + \dots}_{j-y} + (k-j)U + (y+1)2\beta = \sum_{i=0}^j \ell_i + (k-j)U + (j + 1) 2\beta.
\end{align*}
For \DTPRmax, suppose that $c_{\max}(\sigma) < u_{j+1}$.  By the definition of \DTPRmax, we know that accepting any $\ell_i$ helps avoid a switching cost of $-2\beta$ in the worst case. Thus,
\begin{align*}
    \sum_{i=0}^j \ell_i + (k-j) L - 2\beta = \underbrace{u_i + \dots}_{y} + \underbrace{\ell_i + \dots}_{j-y} + (k-j)L - (y+1)2\beta = \sum_{i=0}^j u_i + (k-j)L - (j + 1) 2\beta.
\end{align*}
\end{observation}

With the above observation, for \DTPRmin, one can substitute $u_i - 2\beta$ for each $\ell_i$.  By comparing adjacent terms in Equation~\eqref{eq:balancemin}, standard algebraic manipulations give a closed form for each $u_i$ in terms of $u_1$.  Setting $\frac{kU+2\beta}{k(u_1 - 2\beta) + 2\beta} = \alpha$, we obtain the explicit expression for $u_1$, yielding a closed formula for $\{u_i\}_{i \in [1, k]}$ and $\{\ell_i\}_{i \in [1, k]}$ in Equation~\eqref{eq:minthres}. Considering the balancing rule in Equation~\eqref{eq:balancemin} for the case where $c_{\min}(\sigma) \geq \ell_{k+1}$, it follows that $\ell_{k+1} = L$, and thus $u_{k+1} = L+2\beta$.  By substituting this value into Definition~\ref{def:minthres}, we obtain an explicit expression for $\alpha$ as shown in Equation~\eqref{eq:alpha}.

Conversely, for \DTPRmax, we substitute $\ell_i + 2\beta$ for each $u_i$.  By comparing adjacent terms in Equation~\eqref{eq:balancemax}, standard methods give a closed form for each $\ell_i$ in terms of $\ell_1$.  Setting $\frac{k(\ell_1 + 2\beta) - 2\beta}{kL-2\beta} = \omega$, we obtain the explicit expression for $\ell_1$, yielding the closed formula for $\{\ell_i\}_{i \in [1, k]}$ and $\{u_i\}_{i \in [1, k]}$ in Equation~\eqref{eq:maxthres}.
Considering the balancing rule in Equation \eqref{eq:balancemax} for the case where $c_{\max}(\sigma) \leq u_{k+1}$, it follows that $u_{k+1} = U$, and thus $\ell_{k+1} = U-2\beta$.  By substituting this value into Definition \ref{def:maxthres}, we obtain an explicit expression for $\omega$ as shown in Equation~\eqref{eq:omega}.

\smallskip

%% file: results.tex
We now present competitive results of \texttt{DTPR} for both variants of \OPR and discuss the significance of the results in relation to other algorithms for related problems. Our results for the competitive ratios of \DTPRmin and \DTPRmax are summarized in Theorems~\ref{thm:compPRmin} and~\ref{thm:compPRmax}. We also state the lower bound results for any deterministic online algorithms for \OPRmin and \OPRmax in Theorems~\ref{thm:lowerboundmin} and~\ref{thm:lowerboundmax}. Proofs of the results for \DTPRmin and \DTPRmax are deferred to Section~\ref{sec:analysis} and Appendix~\ref{app:compPRmax}, respectively. Formal proofs of lower bound theorems are given in Appendix~\ref{app:lowerbound}, and a sketch is shown in Section~\ref{sec:prooflowerbound}.
Note that in the competitive results, $W ( x )$ denotes the Lambert $W$ function, i.e., the inverse of $f(x) = xe^x$.  It is well-known that $W(x)$ behaves like $\ln (x)$ \cite{HoorfarHassani:08, Stewart:09}.
We start by presenting our competitive bounds on \DTPRmin and \DTPRmax.

\begin{theorem}\label{thm:compPRmin}
\DTPRmin is an $\alpha$-competitive deterministic algorithm for \texttt{OPR-min}, where $\alpha$ is the unique positive solution of \vspace{-1em}
\begin{align}
\frac{U - L - 2\beta}{U(1 - 1/\alpha) - \left(2\beta - \frac{2\beta}{k} + \frac{2\beta}{k\alpha} \right)} = \left( 1 + \frac{1}{k\alpha} \right)^k. \label{eq:alpha}
\end{align}
\end{theorem}

\begin{theorem}\label{thm:compPRmax}
\DTPRmax is an $\omega$-competitive deterministic algorithm for \texttt{OPR-max}, where $\omega$ is the unique positive solution of \vspace{-1em}
\begin{align}
\frac{U - L - 2\beta}{L(\omega - 1) - 2\beta \left(1 - \frac{1}{k} + \frac{\omega}{k} \right)} = \left( 1 + \frac{\omega}{k} \right)^k. \label{eq:omega}
\end{align}
\end{theorem}

\begin{figure*}[t]
	\minipage{0.48\textwidth}
	\includegraphics[width=\linewidth]{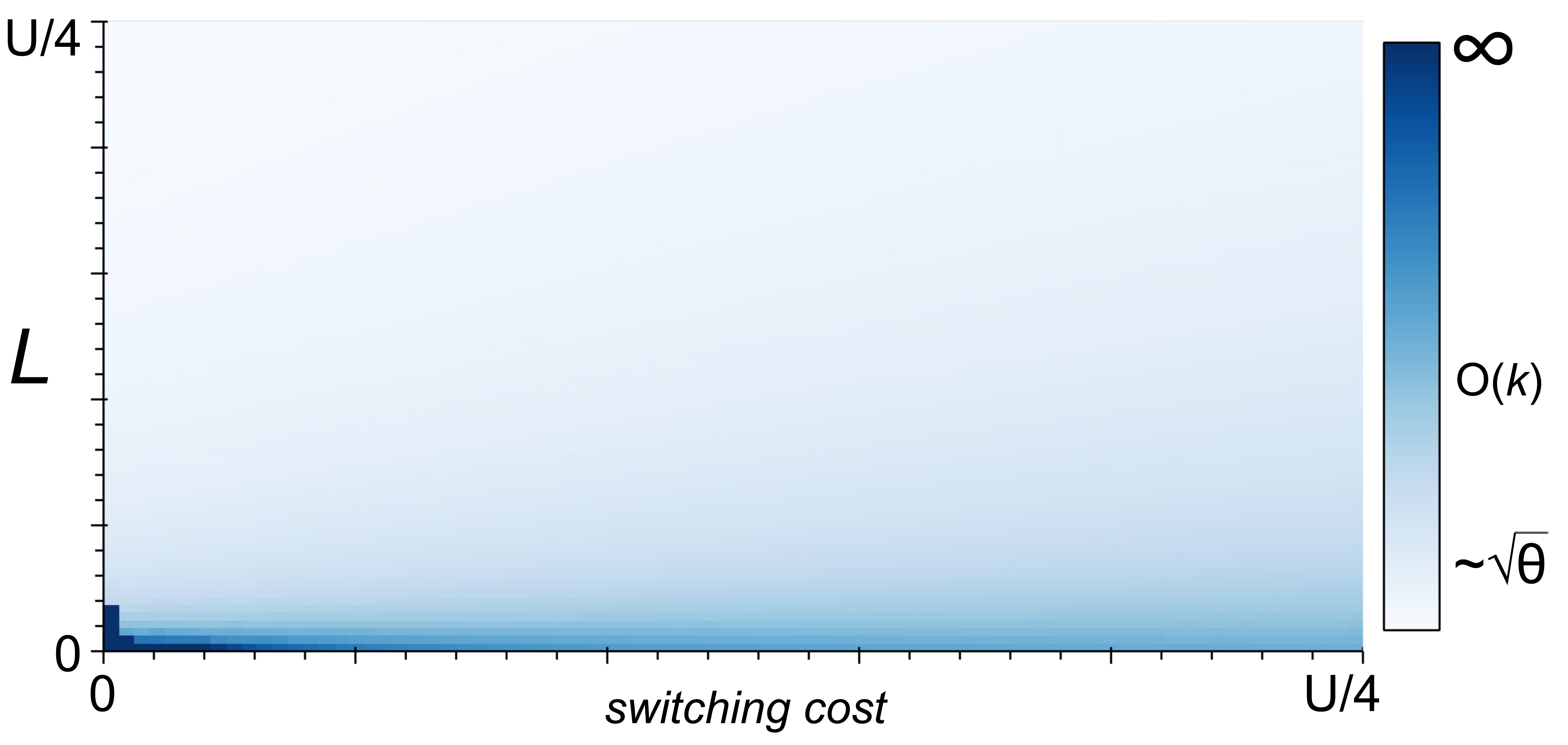}\vspace{-1em}
    \caption{\DTPRmin: Plotting actual values of competitive ratio $\alpha$ for fixed $k \geq 1$, fixed $U > L$, and varying values for $L$ and $\beta$ (switching cost).  Color represents the order of $\alpha$ for a given setting of $\theta$ and $\beta$.}\label{fig:ubplotmin}
	\endminipage\hfill
    \minipage{0.48\textwidth}
	\includegraphics[width=\linewidth]{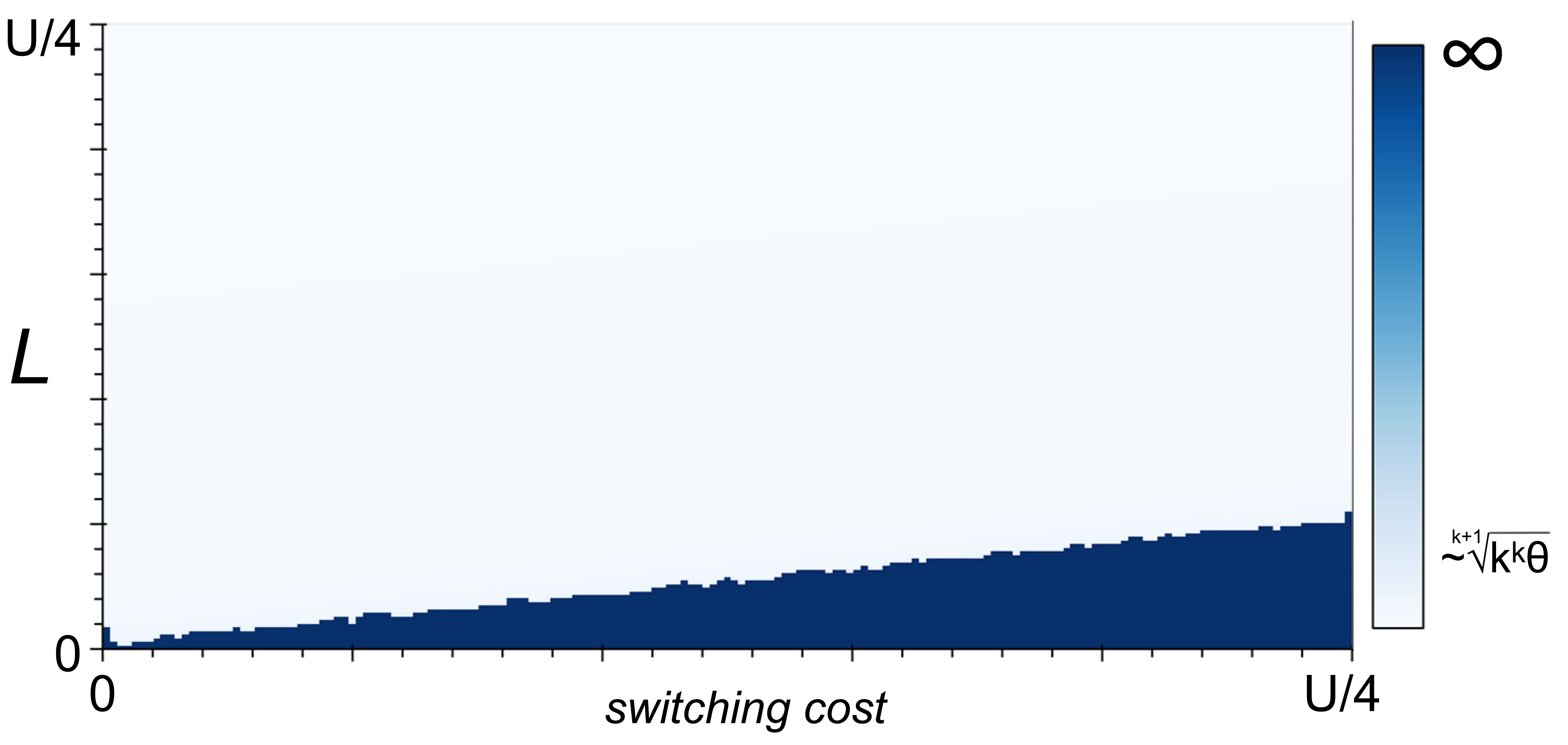}\vspace{-1em}
    \caption{\DTPRmax: Plotting actual values of competitive ratio $\omega$ for fixed $k \geq 1$, fixed $U > L$, and varying values for $L$ and $\beta$ (switching cost). Color represents the order of $\omega$ for a given setting of $\theta$ and $\beta$.}\label{fig:ubplotmax}
	\endminipage\hfill
\end{figure*}

These theorems present upper bounds on the competitive ratios, showing their dependence on the problem parameters. To investigate the behavior of these competitive ratios, in Figures~\ref{fig:ubplotmin} and~\ref{fig:ubplotmax}, we show the competitive ratios of both algorithms as problem parameters are varied. More specifically, in Figure \ref{fig:ubplotmin}, we visualize $\alpha$ as a function of $\beta$ and $L$, where $k$ and $U$ are fixed.  The color (shown as an annotated color bar on the right-hand side of the plot) represents the order of $\alpha$.  If $\beta > 0$ and $L \rightarrow 0$, Figure \ref{fig:ubplotmin} shows that $\alpha$ is roughly $O\left( k \right)$, which we discuss further in Corollary \ref{cor:min}(a).  In Figure \ref{fig:ubplotmax}, we visualize $\omega$ as a function of $\beta$ and $L$, where $k$ and $U$ are fixed.  The color represents the order of $\omega$.  In the dark blue region of the plot, Figure \ref{fig:ubplotmax} shows that $\omega \rightarrow \infty$ when $b \rightarrow k$, which provides insight into the extreme case for switching cost when $\beta \gtrsim \frac{kL}{2}$.

To obtain additional insight into the form of the competitive ratios in Theorems \ref{thm:compPRmin} and \ref{thm:compPRmax}, we present the following corollaries for two asymptotic regimes of interest: \texttt{REGIME-1} captures the order of the competitive ratio when $k$ is fixed and $\alpha$ or $\omega$ are sufficiently large, and \texttt{REGIME-2} captures the order of the competitive ratio when $k \rightarrow \infty$.

\begin{corollary}
\label{cor:min}
\textbf{(a)} For \texttt{REGIME-1}, with fixed $k \geq 1$ and $\beta \in (0, \frac{U-L}{2})$, the competitive ratio of \DTPRmin is
$$\alpha \thicksim \frac{k\beta}{kL+2\beta} + \sqrt{ \frac{k^2 LU + 2kL\beta + 2kU\beta + 4\beta^2 + k^2\beta^2}{k^2L^2 + 4kL\beta + 4\beta^2} }, \quad \textit{and $\alpha \thicksim O \left(k \right) $ for $L \rightarrow 0$.}$$

\textbf{(b)} Furthermore, for \texttt{REGIME-2}, with $k \rightarrow \infty$ and $c = \frac{2\beta}{U}, c \in (0, \frac{U-L}{U})$, the competitive ratio of \DTPRmin is 
$$\alpha \thicksim \left[ W \left( \frac{\left( c + \frac{1}{\theta} -1 \right) e^c}{e} \right) - c + 1 \right]^{-1}.$$
\end{corollary}

\begin{corollary}
\label{cor:max}
\textbf{(a)} For \texttt{REGIME-1}, with fixed $k \geq 1$ and $b = \frac{2\beta}{L}, b \in (0, k)$, the competitive ratio of \DTPRmax is 
$$\omega \thicksim O \left( \sqrt[k+1]{k^k \frac{k\theta}{k - b}} \right),$$

and \textbf{(b)} for \texttt{REGIME-2}, with $k \rightarrow \infty$ and $b = \frac{2\beta}{L}, b \in (0, k)$, the competitive ratio of \DTPRmax is
$$\omega \thicksim W \left( \frac{\theta - 1 - a}{e^{1+b}} \right) + 1 + b.$$
\end{corollary}

Corollary \ref{cor:min}(a) contextualizes the behavior of $\alpha$ (the competitive ratio of \DTPRmin) in the most relevant \OPRmin setting (when $\beta \in (0, \frac{U-L}{2})$).  Let us also briefly discuss the other cases for the switching cost $\beta$, and why this interval makes sense.
When $\beta > \frac{U-L}{2}$, the switching cost is large enough such that \texttt{OPT} only incurs a switching cost of $2\beta$.  In this regime, $\alpha$ does not fully capture the competitive ratio of \DTPRmin, since every value in the threshold family $\{u_i\}_{i \in [1, k]}$ is at least $U$; in other words, whenever the algorithm begins accepting prices, it will accept $k$ prices in a single continuous segment, incurring minimal switching cost of $2\beta$.  As $\beta \rightarrow \infty$, the competitive ratio of \DTPRmin~approaches $1$.

Conversely, Corollary \ref{cor:max}(a) contextualizes the behavior of $\omega$ in the most relevant \OPRmax setting (when $\beta \in (0, \frac{kL}{2})$), but we also discuss the other cases for the switching cost $\beta$, and why this interval makes sense.
When $\beta \geq \frac{kL}{2}$, the switching cost is too large, and the competitive ratio may become unbounded.  Note that this is shown explicitly in Figure \ref{fig:ubplotmax}.  Consider an adversarial sequence which forces any \OPRmax algorithm to accept $k$ prices with value $L$ at the end of the sequence.  On such a sequence, even a player which incurs the minimum switching cost of $2\beta$ achieves zero or negative profit of $kL - 2\beta \leq 0$, and this is not well-defined.

Next, to begin to investigate the tightness of Theorems \ref{thm:compPRmin} and \ref{thm:compPRmax}, it is interesting to consider special cases that correspond to models studied in previous work.  In particular, when $\beta =0$, i.e., there is no switching cost, $\OPR$ degenerates to the $k$-search problem~\cite{Lorenz:08}.  For fixed $k \geq 1$ and $\theta \rightarrow \infty$, the optimal competitive ratios shown by \cite{Lorenz:08} are $\sqrt{\theta/2}$ for $k$-min, and $\sqrt[k+1]{k^k \theta}$ for $k$-max (see Section~\ref{sec:OTA}). Both versions of \DTPR exactly recover the optimal $k$-search algorithms~\cite{Lorenz:08}.\footnote{To see this, note that by eliminating all $\beta$ terms from Equations~\eqref{eq:alpha} and \eqref{eq:omega}, we exactly recover Equations~\eqref{eq:kminalpha} and \eqref{eq:kmaxomega}, which are the definitions of the $k$-search algorithms. When $\theta \rightarrow \infty$ as $L \rightarrow 0$, \DTPRmin and \DTPRmax match each $k$-search result exactly when $\beta = 0$.  In Corollaries~\ref{cor:min}(b) and \ref{cor:max}(b), \DTPRmin and \DTPRmax also match each $k$ search result exactly when $k \rightarrow \infty$ and $\beta = 0$. (See Sec. \ref{sec:OTA})}
Figure~\ref{fig:ubplotmin} shows that if $\beta = 0$ and $L \rightarrow 0$, then $\alpha \rightarrow \infty$, which matches the $k$-min result of $\sqrt{\theta / 2} \thicksim \infty$.  Similarly, Figure~\ref{fig:ubplotmax} shows that if $\beta = 0$ and $L \rightarrow 0$, then $\omega \rightarrow \infty$, which matches the $k$-max result of  $\sqrt[k+1]{k^k \theta} \thicksim \infty$. 

More generally, one can ask if the competitive ratios of \DTPR can be improved upon by other online algorithms outside of the special case of $k$-search.  Our next set of results highlights that no improvement is possible, i.e., that \DTPRmin and \DTPRmax maintain the optimal competitive ratios possible for any deterministic online algorithm for \OPR.

\begin{theorem}\label{thm:lowerboundmin}
Let $k \geq 1$, $\theta \geq 1$, and $\beta \in (0, \frac{U-L}{2})$.  Then $\alpha$ given by Equation~\eqref{eq:alpha} is the best competitive ratio that a deterministic online algorithm for \OPRmin can achieve.
\end{theorem}

\begin{theorem}\label{thm:lowerboundmax}
Let $k \geq 1$, $\theta \geq 1$, and $\beta \in (0, \frac{kL}{2})$.  Then $\omega$ given by Equation (\ref{eq:omega}) is the best competitive ratio that a deterministic online algorithm for \OPRmax can achieve.
\end{theorem}

By combining Theorems~\ref{thm:compPRmin} and~\ref{thm:compPRmax} with Theorems~\ref{thm:lowerboundmin}  and~\ref{thm:lowerboundmax}, these results imply that the competitive ratios of \DTPRmin and \DTPRmax are optimal for \OPRmin and \OPRmax.

Finally, it is interesting to contrast the upper and lower bounds for \OPR with those for $k$-search, since the contrast highlights the impact of switching costs. In \OPRmin with $\beta > 0$, \DTPRmin \textit{improves} on existing optimal results for $k$-min search, particularly in the case where $L$ approaches~$0$ (i.e., $\theta \rightarrow \infty$).  Since Theorem \ref{thm:lowerboundmin} implies that \DTPRmin is optimal, this shows that the addition of switching cost in \OPRmin enables an online algorithm to achieve a better competitive ratio compared to $k$-min search, which is a surprising result.
In contrast, for \OPRmax with $\beta > 0$, \DTPRmax's competitive bounds are \textit{worse} than existing results for $k$-max search, particularly for large $\beta$.  Since Theorem \ref{thm:lowerboundmax} implies that \DTPRmax is optimal, this suggests that \OPRmax is fundamentally a \textit{more difficult} problem compared to $k$-max search.

%% file: analysis.tex
We now prove the results described in the previous section. In Section~\ref{sec:proof4}, we prove the \DTPRmin results presented in Theorem~\ref{thm:compPRmin} and Corollary~\ref{cor:min}. In Section~\ref{sec:prooflowerbound}, we provide a proof sketch for the lower bound results in Theorems~\ref{thm:lowerboundmin} and~\ref{thm:lowerboundmax}, and defer the formal proofs to Appendix~\ref{app:lowerbound}.  The competitive results for \DTPRmax in Theorem~\ref{thm:compPRmax} and Corollary~\ref{cor:max} are deferred to Appendix~\ref{app:compPRmax}.

\subsection{Competitive Results for \DTPRmin} \label{sec:proof4}

We begin by proving Theorem~\ref{thm:compPRmin} and Corollary~\ref{cor:min}. The key novelty in the proof of the main competitive results (Theorems \ref{thm:compPRmin} and \ref{thm:compPRmax}) lies in our effort to derive two threshold functions and balance the competitive ratio in several worst-case instances with respect to these thresholds, as outlined in Section~\ref{sec:dtpr}.

\begin{proof}[Proof of Theorem \ref{thm:compPRmin}] 
For $0 \leq j \leq k$, let $\mathcal{S}_j \subseteq \mathcal{S}$ be the sets of \OPRmin price sequences for which \DTPRmin accepts exactly $j$ prices (excluding the $k-j$ prices it is forced to accept at the end of the sequence). Then, all of the possible price sequences for \OPRmin are represented by $\mathcal{S} = \bigcup_{j=0}^k \mathcal{S}_j$.  Also, recall that by definition, $\ell_{k+1} = L$.  Let $\epsilon > 0$ be a fixed constant, and define the following two price sequences $\sigma_j$ and $\rho_j$:
$$\forall j \in [2, k] : \sigma_j = \ell_1, u_2, \dots, u_j, U, \underbrace{\ell_{j+1} + \epsilon, \dots, \ell_{j+1} + \epsilon}_{k}, \underbrace{U, U, \dots, U}_{k}.$$
$$\forall j \in [2, k] : \rho_j = \ell_1, U, \ell_2, U, \dots, U, \ell_j, U, \underbrace{\ell_{j+1} + \epsilon, \dots, \ell_{j+1} + \epsilon}_{k}, \underbrace{U, U, \dots, U}_{k}.$$
There are two special cases for $j=0$ and $j=1$.  For $j=0$, we have that $\sigma_0 = \rho_0$, and this sequence simply consists of $\ell_{1} + \epsilon$ repeated $k$ times, followed by $U$ repeated $k$ times.  For $j = 1$, we also have that $\sigma_1 = \rho_1$, and this sequence consists of one price with value $\ell_1$ and one price with value $U$, followed by $\ell_{2} + \epsilon$ repeated $k$ times and $U$ repeated $k$ times.
\smallskip

Observe that as $\epsilon \rightarrow 0$, $\sigma_j$ and $\rho_j$ are sequences yielding the worst-case ratios in $\mathcal{S}_j$, as \DTPRmin\; is forced to accept $(k-j)$ worst-case $U$ values at the end of the sequence, and each accepted value is exactly equal to the corresponding threshold.

Note that $\sigma_j$ and $\rho_j$ also represent two extreme possibilities for the additive switching cost.  In $\sigma_j$, \DTPRmin only switches twice, but it mostly accepts values $u_i$.  In $\rho_j$, \DTPRmin~must switch $j + 1$ times because there are many intermediate $U$ values, but it only accepts values $\ell_i$.  

In the worst case, we have 
$$\frac{\DTPRmin(\sigma_j)}{\texttt{OPT}(\sigma_j)} = \frac{\DTPRmin(\rho_j)}{\texttt{OPT}(\rho_j)}.$$ 
Also, the optimal solutions for both sequences are lower bounded by the same quantity: $k c_{\min}(\sigma_j) + 2\beta = k c_{\min}(\rho_j) + 2\beta$.
For any sequence $s$ in $\mathcal{S}_j$, we have that $c_{\min}(s) > \ell_{j+1}$, so $\texttt{OPT}(\rho_j) = \texttt{OPT}(\sigma_j) \leq k \ell_{j+1} + 2\beta$.\\
By definition of the threshold families $\{\ell_i\}_{i \in [1, k]}$ and $\{u_i\}_{i \in [1, k]}$, we know that $\sum_{i=1}^j~\ell_i~+~j2\beta~=~\sum_{i=1}^j u_i$ for any value $j \geq 2$:
\begin{align*}
\DTPRmin(\rho_j) = \left( \sum_{i=1}^j \ell_i + (k-j)U + (j+1)2\beta \right) = \left( \ell_1 + \sum_{i=2}^j u_i + (k-j)U + 4\beta \right) = \DTPRmin(\sigma_j).
\end{align*}
Note that whenever $j < 2$, we have that $\sigma_0 = \rho_0$, and $\sigma_1 = \rho_1$.  Thus, $\DTPRmin(\rho_j) = \DTPRmin(\sigma_j)$ holds for any value of $j$.
By definition of $\ell_1$, we simplify $\ell_1 + \sum_{i=2}^j u_i + (k-j)U + 4\beta$ to\\ $\sum_{i=1}^j u_i + (k-j)U + 2\beta$.  Then, for any sequence $s \in \mathcal{S}_j$, we have the following:
\begin{align}
\frac{\DTPRmin(s)}{\texttt{OPT}(s)} \leq \frac{\DTPRmin(\sigma_j)}{\texttt{OPT}(\sigma_j)} = \frac{\DTPRmin(\rho_j)}{\texttt{OPT}(\rho_j)} \leq \frac{\sum_{i=1}^j u_i + (k-j)U + 2\beta}{k\ell_{j+1} + 2\beta}.
\end{align}
Before proceeding to the next step, we use an intermediate result stated in the following lemma with a proof given in Appendix \ref{appendix:proofs}.
\begin{lemma} \label{lem:intermedStepMin}
For any $0 \leq j \leq k$, by definition of $\{\ell_i\}_{i \in [1, k]}$ and $\{u_i\}_{i \in [1, k]}$,
\begin{align*}
\sum_{i=1}^j u_i + (k-j)U + 2\beta \leq \alpha \cdot (k\ell_{j+1} + 2\beta).
\end{align*}
\end{lemma}
\noindent For $\epsilon \rightarrow 0$, the competitive ratio $\DTPRmin/\texttt{OPT}$ is exactly $\alpha$:
\begin{align*}
\forall 0 \leq j \leq k: \;\;\; \frac{\DTPRmin(\sigma_j)}{\texttt{OPT}(\sigma_j)} = \frac{\sum_{i=1}^j u_i + (k-j)U + 2\beta}{k\ell_{j+1} + 2\beta} = \alpha,
\end{align*}
and thus for any sequence $s \in \mathcal{S}$,
\begin{align*}
\forall s \in \mathcal{S}: \;\;\; \frac{\DTPRmin(s)}{k c_{\min}(s) + 2\beta} \leq \alpha.
\end{align*}
Since $\texttt{OPT}(s) \geq k c_{\min}(s) + 2\beta$ for any sequence $s$, this implies that $\DTPRmin$ is $\alpha$-competitive.
\end{proof}

\begin{proof}[Proof of Corollary~\ref{cor:min}]
To show part \textbf{(a)} for \texttt{REGIME-1}, with fixed $k \geq 1$, observe that we can expand the right-hand side of Equation~\eqref{eq:alpha} using the binomial theorem to obtain the following:
\begin{align*}
\frac{U - L - 2\beta}{U \left( 1 - \frac{1}{\alpha} \right) - 2\beta \left(1 - \frac{1}{k} + \frac{1}{k\alpha} \right)} = 1 + \frac{1}{\alpha} + \Theta \left( \alpha^{-2} \right).
\end{align*}
Next, observe that $\alpha^\star$ solving the following expression satisfies $\alpha^\star \geq \alpha \;\; \forall k : k \geq 1$, (i.e. $\alpha^\star$ is an upper bound of $\alpha$): \vspace{-1em}
\begin{align*}
\frac{U - L - 2\beta}{U \left( 1 - \frac{1}{\alpha^\star} \right) - 2\beta \left(1 - \frac{1}{k} + \frac{1}{k\alpha^\star} \right)} = 1 + \frac{1}{\alpha^\star}.
\end{align*}
\\
By solving the above for $\alpha^\star$, we obtain 
$$\alpha \thicksim \alpha^\star = \frac{k\beta}{kL + 2\beta} + \sqrt{\frac{k^2 LU + 2kL\beta + 2kU\beta + 4\beta^2 +k^2\beta^2}{k^2 L^2 + 4kL\beta + 4\beta^2}}.$$
Last, note that as $L \rightarrow 0$, we obtain the following result: $\alpha \thicksim \frac{k}{2} + \sqrt{\frac{kU}{2\beta} + 1 + \frac{k^2}{4}} \approx O\left( k \right)$.\\

To show part \textbf{(b)} for \texttt{REGIME-2}, we first observe that the right-hand side of Equation~\ref{eq:alpha} can be approximated as $\left( 1 + \frac{1}{k\alpha} \right)^k~\approx~e^{1/\alpha}$ when $k \rightarrow \infty$.  Then by taking limits on both sides, we obtain the following:
\begin{align*}
\frac{U - L - 2\beta}{U \left( 1 - \frac{1}{\alpha} \right) - 2\beta\left(1 \right)} = e^{1/\alpha}.
\end{align*}
For simplification purposes, let $\beta = cU/2$, where $c$ is a small constant on the interval $\left(0, \frac{U-L}{U} \right)$.\\  We then obtain the following:
\begin{align*}
\frac{U - L - cU}{U \left( 1 - \frac{1}{\alpha} \right) - cU} = e^{1/\alpha} \Longrightarrow L/U + c - 1 = \left( \frac{1}{\alpha} + c - 1 \right) e^{1/\alpha}.
\end{align*}
By definition of Lambert $W$ function, solving this equation for $\alpha$ obtains the result in Corollary~\ref{cor:min}(b).
\end{proof}

\subsection{Lower Bound Analysis: Proof Sketch for Theorems \ref{thm:lowerboundmin} and \ref{thm:lowerboundmax}} \label{sec:prooflowerbound}

Here we present a proof sketch for the lower bound construction that is used to prove both Theorems \ref{thm:lowerboundmin} and \ref{thm:lowerboundmax}.  We show how to formalize it in the case of Theorem \ref{thm:lowerboundmin} in Appendix~\ref{app:lowerboundmin}, and in the case of Theorem \ref{thm:lowerboundmax} in Appendix \ref{app:lowerboundmax}.

Suppose that $\texttt{ALG}$ is a deterministic online algorithm for \OPR.  The lower bound proofs for both \OPRmin and \OPRmax leverage the same instance, where $\texttt{ALG}$ plays against an \textit{adaptive adversary}. 

To describe the instance, we first need some preliminaries.  Define a sequence of prices $\mathcal{T}_1, \dots, \mathcal{T}_k$, which are the prices the adversary will present to $\texttt{ALG}$.  The ``worst-case value'' that $\texttt{ALG}$ can encounter is defined based on the problem variant.  Since we assume that prices are bounded on the interval $[L, U]$, these values are $U$ for \OPRmin, and $L$ for \OPRmax.

The adversary begins by presenting $\mathcal{T}_1$ to $\texttt{ALG}$, at most $k$ times or until $\texttt{ALG}$ accepts it.  If $\texttt{ALG}$ never accepts $\mathcal{T}_1$, the adversary presents the worst-case value at least $k$ times for the remainder of the sequence.  In the formal proof, we show that this case causes $\texttt{ALG}$ to achieve a competitive ratio of at least $\alpha$ for \OPRmin, or at least $\omega$ for \OPRmax.

If $\texttt{ALG}$ does accept $\mathcal{T}_1$, the adversary continues the sequence by presenting the worst-case value to $\texttt{ALG}$, at most $k$ times \textit{or until $\texttt{ALG}$ switches to reject it}.  This essentially forces $\texttt{ALG}$ to switch immediately after accepting $\mathcal{T}_1$.  In the formal proof, we show that any algorithm which does not switch away immediately achieves a competitive ratio worse than $\alpha$ and $\omega$ for \OPRmin and \OPRmax.

After $\texttt{ALG}$ has switched away, the adversary continues the sequence by presenting $\mathcal{T}_2$ to $\texttt{ALG}$ at most $k$ times or until $\texttt{ALG}$ accepts it.  Again, if $\texttt{ALG}$ never accepts $\mathcal{T}_2$, the adversary presents the worst-case value at least $k$ times for the remainder, and $\texttt{ALG}$ cannot do better than $\alpha$ or $\omega$.

The adversary continues in this fashion, presenting each $\mathcal{T}_i$ at most $k$ times (or until $\texttt{ALG}$ accepts it and the adversary forces $\texttt{ALG}$ to switch away immediately afterward).  Whenever $\texttt{ALG}$ does not accept some $\mathcal{T}_i$ after it is presented $k$ times, the adversary sends the price to the worst-case value for the remainder of the sequence.  If $\texttt{ALG}$ accepts $k$ prices before the end of the sequence, the adversary concludes by presenting the best-case value ($L$ for \OPRmin, $U$ for \OPRmax) at least $k$ times.

In the formal proofs presented in Appendix~\ref{app:lowerbound}, we show that \textit{any} deterministic strategy  that $\texttt{ALG}$ uses to accept prices on this sequence 
achieves a competitive ratio of at least $\alpha$ for \OPRmin, and at least $\omega$ for \OPRmax.

%% file: experiments.tex
We now present experimental results for the \DTPR algorithms in the context of the carbon-aware temporal workload shifting problem. We evaluate \DTPRmin~(and \DTPRmax in Appendix~\ref{appendix:maxExp}) as compared to existing algorithms from the literature that have been adapted for \OPR.

\subsection{Experimental Setup}  \label{sec:expsetup}

We consider a carbon-aware load shifting system that operates on a hypothetical data center. An algorithm is given a deferrable and interruptible job that takes $k$ time slots to complete, along with a deadline $T \geq k$, such that the job must be completed at most $T$ slots after its arrival.  The objective is to selectively run units of the job such that the total carbon emissions are minimized while still completing the job before its deadline.

For the minimization variant (\OPRmin) of the experiments, we consider \textit{carbon emissions intensities}, as the price values.  At each time step $t$, the electricity supply has a carbon intensity $c_t$, i.e., if the job is being processed during the time step $t$ ($x_t = 1$), the data center's carbon emissions during that time step are proportional to $c_t$.  If the job is \textit{not} being processed during the time step $t$ ($x_t = 0$), we assume for simplicity that carbon emissions in the idle state are negligible and essentially $0$.
To model the combined computational overhead of interrupting, checkpointing, and restarting the job, the algorithm incurs a fixed switching cost of $\beta$ whenever $x_{t-1} \not = x_t$, whose values are selected relative to the price values.

\paragraph{Carbon data traces}
We use real-world carbon traces from Electricity Maps~\cite{electricity-map}, which provide time-series information about the \textit{average carbon emissions intensity} of the electric grid.
We use traces from three different regions: the Pacific Northwest of the U.S., New Zealand, and Ontario, Canada.  
The data is provided at an hourly granularity and includes the current average carbon emissions intensity in grams of CO$_2$ equivalent per kilowatt-hour (gCO$_2$eq/kWh), and the percentage of electricity being supplied from carbon-free sources. In Figure~\ref{fig:traceVis} (in Appendix~\ref{appendix:maxExp}), we plot three representative actual traces for carbon intensity over time for a 96-hour period in each region.

\paragraph{Parameter settings}
We test for time horizons ($T$) of 48 hours, 72 hours, and 96 hours.  The chosen time horizon represents the time at which the job with length $k$ must be completed.  As is given in the carbon trace data, we consider time slots of one hour.  

The online algorithms we use in experiments take $L$ and $U$ as parameters for their threshold functions. To set these parameters, we examine the entire carbon trace for the current location.  For the Pacific NW trace and the Ontario trace, these values represent lower and upper bounds of the carbon intensity values for a full year.  For the New Zealand trace, these values are a lower and upper bound for the values during a month of data, which is reflected by a smaller fluctuation ratio.
We set $L$ and $U$ to be the minimum and maximum observed carbon intensity over the entire trace.  
\begin{table}[t]
\caption{Summary of carbon trace data sets} \label{tab:characteristics}
\begin{tabular}{|l|l|l|l|}
\hline
Location & \textbf{PNW, U.S.} & \textbf{New Zealand} & \textbf{Ontario, Canada} \\ \hline\hline
Number of Data Points    & 10,144         & 1,324     & 17,898 \\ \hline
Max. \footnotesize Carbon Intensity ($U$) \normalsize & 648 \footnotesize gCO$_2$eq/kWh \normalsize & 165 \footnotesize gCO$_2$eq/kWh \normalsize & 181 \footnotesize gCO$_2$eq/kWh \normalsize \\ \hline
Min. \footnotesize Carbon Intensity ($L$) \normalsize & 18 \footnotesize gCO$_2$eq/kWh \normalsize & 54 \footnotesize gCO$_2$eq/kWh \normalsize & 15 \footnotesize gCO$_2$eq/kWh \normalsize \\ \hline\hline
Duration (mm/dd/yy) & 04/20/22 - 12/06/22 & 10/19/21 - 11/16/21 & 10/19/21 - 12/06/22\\ \hline
\end{tabular}
\vspace{-1.5em}
\end{table}

To generate each input sequence, a contiguous segment of size $T$ is randomly sampled from the given carbon trace.  
In a few experiments, we simulate greater \textit{volatility} over time by ``scaling up'' each price's deviation from the mean.  First, we compute the average value over the entire sequence.  Next, we compute the difference between each price and this average.  Each of these differences is scaled by a \textit{noise factor} of $m \geq 1$.  Finally, new carbon values are computed by summing each scaled difference with the average.  If $m = 1$, we recover the same sequence, and if $m > 1$, any deviation from the mean is proportionately amplified.  Any values which become negative after applying this transformation are truncated to $0$.  This technique allows us to evaluate algorithms under different levels of volatility.  Performance in the presence of greater carbon volatility is important, as on-site renewable generation is seeing greater adoption as a supplementary power source for data centers~\cite{radovanovic2022carbon, acun2022holistic}.

\paragraph{Benchmark algorithms}

To evaluate the performance of \DTPR, we use a dynamic programming approach to calculate the offline optimal solution for each given sequence and objective, which allows us to report the empirical competitive ratio for each tested algorithm. 
We compare \texttt{DTPR} against two categories of benchmark algorithms, which are summarized in Table \ref{tab:baselineAlgos}.  

The first category of benchmark algorithms is \textit{carbon-agnostic} algorithms, which run the jobs during the first $k$ time slots in order, i.e., accepting prices $c_1, \dots, c_k$.  This approach incurs the minimal switching cost of $2\beta$, because it does not interrupt the job while it is being processed.  The carbon-agnostic approach simulates the behavior of a scheduler that runs the job to completion as soon as it is submitted, without any focus on reducing carbon emissions.  Note that the performance of this approach significantly varies based on the randomly selected sequence, since it will perform well if low-carbon electricity is available in the first few slots, and will perform poorly if the first few slots are high-carbon.

We also compare \texttt{DTPR} against \textit{switching-cost-agnostic} algorithms, which only consider carbon cost.  We have two algorithms of this type, each drawing from existing online search methods in the literature.  Although they do not consider the switching cost in their design, they still incur a switching cost whenever their decision in adjacent time slots differs.

The first such algorithm is a \textit{constant threshold algorithm}, which uses the $\sqrt{UL}$ threshold value first presented for online search in \cite{ElYaniv:01}.  In our minimization experiments, this algorithm runs the workload during the first $k$ time slots where the carbon intensity is at most $\sqrt{UL}$.  

The other switching-cost-agnostic algorithm tested is the $k$-search algorithm shown by \cite{Lorenz:08} and described in Section \ref{sec:OTA}.  The $k$-min search algorithm chooses to run the $i$th hour of the job during the first time slot where the carbon intensity is at most $\Phi_i$.  

\begin{table}[]
\caption{Summary of algorithms tested in our experiments} \label{tab:baselineAlgos}
\begin{tabular}{|l|l|l|l|}
\hline
\footnotesize \textbf{Algorithm} & \footnotesize \textbf{Carbon-aware} & \footnotesize \textbf{Switching-aware} & \footnotesize \textbf{Description} \\ \hline\hline
\texttt{OPT} (offline)      & \texttt{YES}          & \texttt{YES}                  & \footnotesize Optimal offline solution \normalsize               \\ \hline
Carbon-Agnostic    & \texttt{NO}           & \texttt{YES}                  & \footnotesize Runs job in the first $k$ time slots \normalsize              \\ \hline
Const. Threshold & \texttt{YES}          & \texttt{NO}                   & \footnotesize Runs job if  carbon meets threshold $\sqrt{UL}$ \normalsize               \\ \hline
$k$-search           & \texttt{YES}          & \texttt{NO}                   & \footnotesize Runs $i$th slot of job if carbon meets threshold $\Phi_i$         \\ \hline
\DTPR           & \texttt{YES}          & \texttt{YES}                   & \footnotesize This work (algorithms proposed in Section~\ref{sec:dtpr})\normalsize              \\ \hline 
\end{tabular}
\end{table}

\subsection{Experimental Results} \label{sec:expresults}

We now present our experimental results.  
Our focus is on the empirical competitive ratio (a lower competitive ratio is better).  We report the performance of all algorithms for each experimental setting, in each tested region. Throughout the minimization experiments, we observe that \DTPRmin outperforms the benchmark algorithms.  The 95th percentile worst-case empirical competitive ratio achieved by \DTPRmin is a $48.2$\% improvement on the carbon-agnostic method, a $15.6$\% improvement on the $k$-min search algorithm, and a $14.4$\% improvement on the constant threshold algorithm.

\begin{figure*}[t]
    \begin{center}
        \minipage{0.7\textwidth}
        \includegraphics[width=\linewidth]{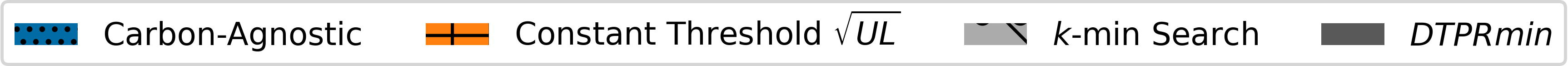}
        \endminipage\hfill\\
    	\minipage{0.33\textwidth}
    	\includegraphics[width=\linewidth]{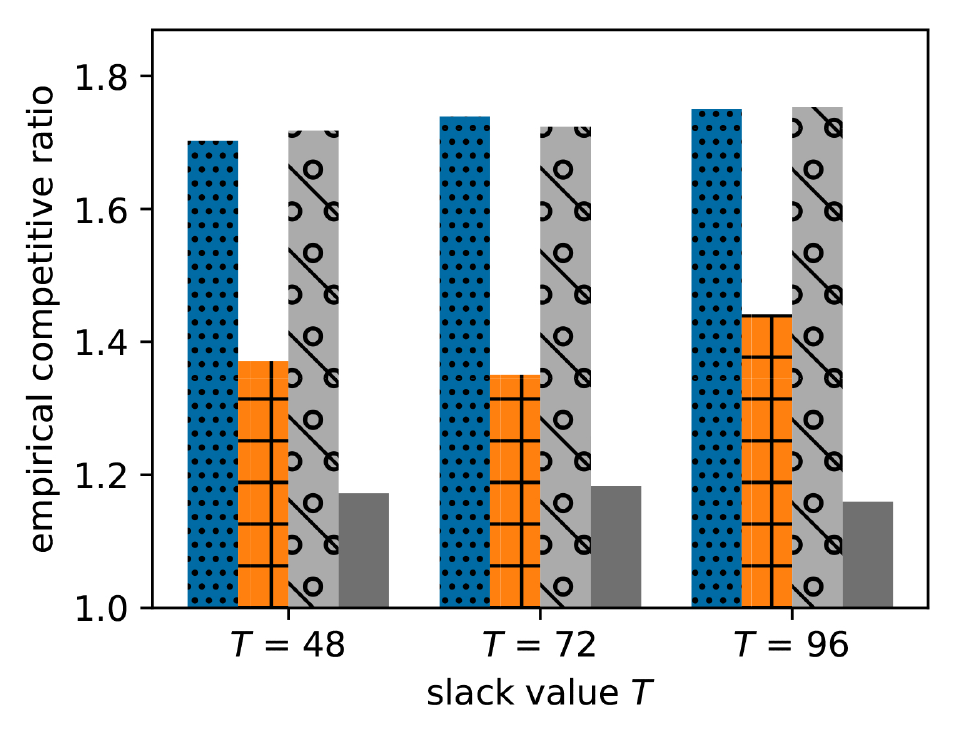}\vspace{-1.2em}
        \caption*{(a) Ontario, Canada}
    	\endminipage\hfill
    	\minipage{0.33\textwidth}
    	\includegraphics[width=\linewidth]{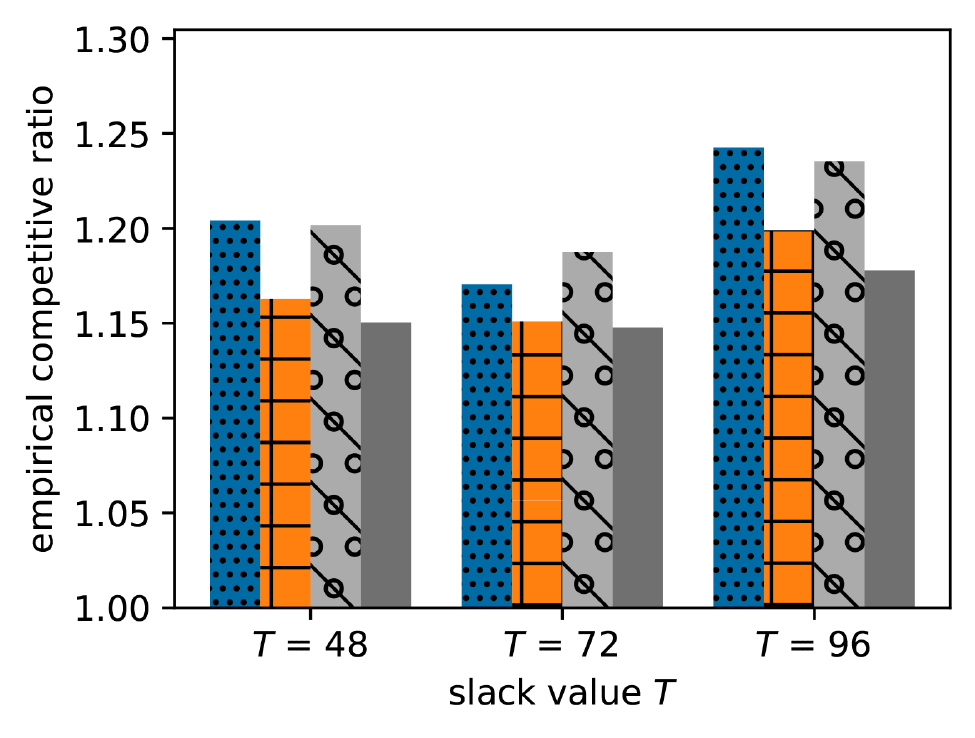}\vspace{-1.2em}
        \caption*{(b) U.S. Pacific Northwest}
    	\endminipage\hfill
        \minipage{0.33\textwidth}
    	\includegraphics[width=\linewidth]{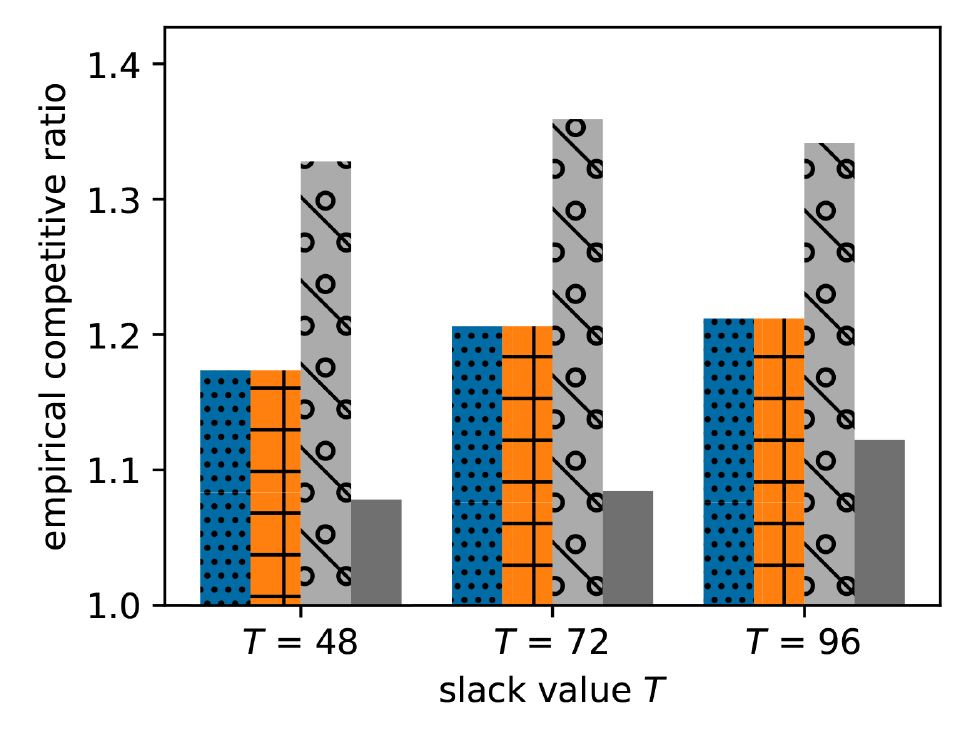}\vspace{-1.2em}
        \caption*{(c) New Zealand}
    	\endminipage\hfill
        \caption{Experiments for three distinct \textit{slack values}, where $T \in \{48, 72, 96\}$.\\ (a): Ontario, Canada carbon trace, with $\theta = 12.0\bar{6}$  \;\; (b): U.S. Pacific Northwest carbon trace, with $\theta = 36$\\  (c): New Zealand carbon trace, with $\theta = 3.0\bar{5}$}
        \label{fig:slack}
    \end{center}
\end{figure*}

\begin{figure*}[t]
    \begin{center}
    \minipage{0.7\textwidth}
    \includegraphics[width=\linewidth]{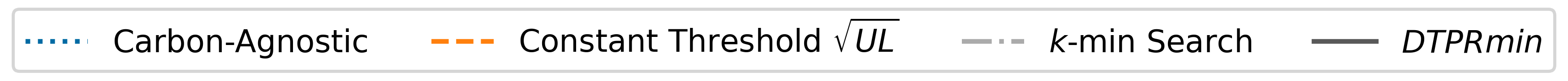}
    \endminipage\hfill\\
	\minipage{0.24\textwidth}
	\includegraphics[width=\linewidth]{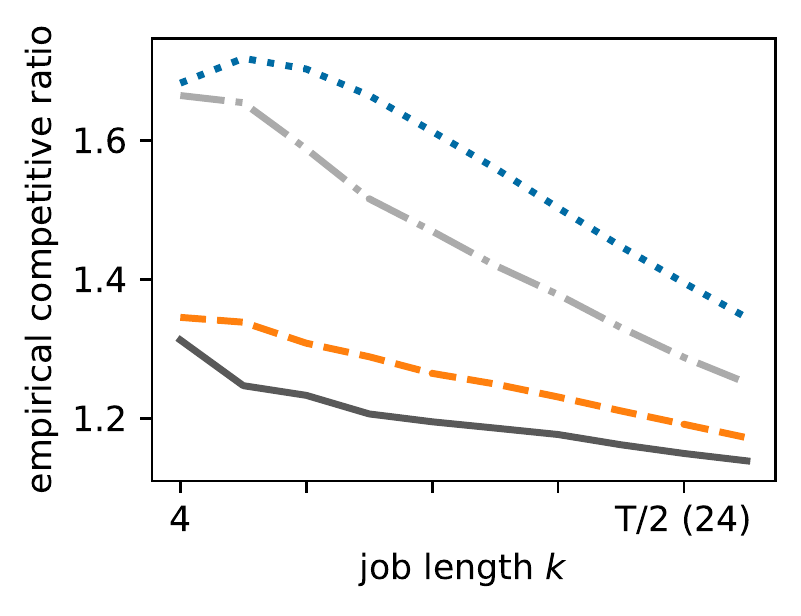}\vspace{-1.2em}
    \caption*{(a) Changing $k$}
	\endminipage\hfill
	\minipage{0.24\textwidth}
	\includegraphics[width=\linewidth]{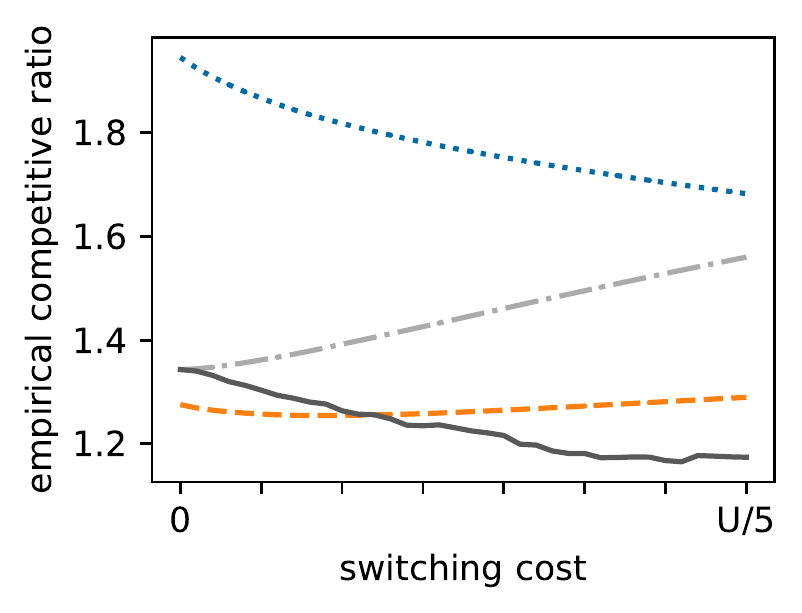}\vspace{-1.2em}
    \caption*{(b) Changing $\beta$}
	\endminipage\hfill
    \minipage{0.24\textwidth}
	\includegraphics[width=\linewidth]{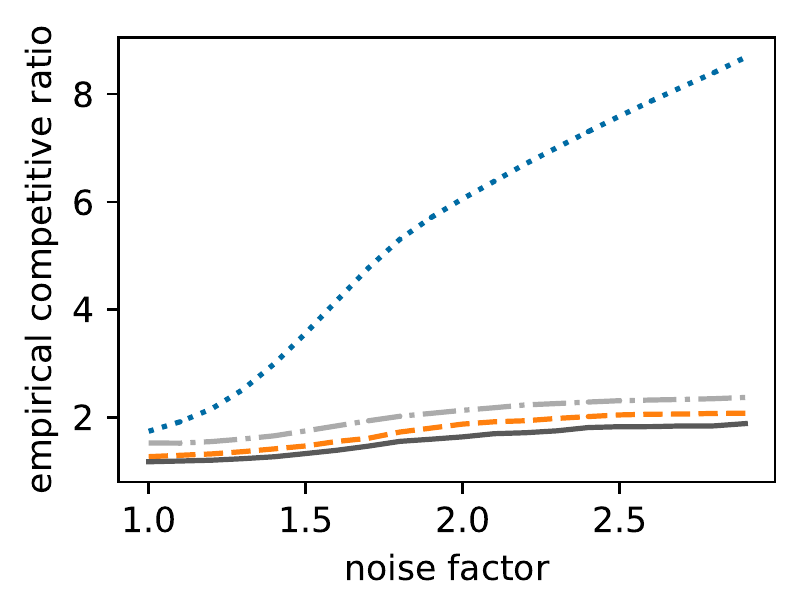}\vspace{-1.2em}
    \caption*{(c) Changing volatility}
	\endminipage\hfill
    \minipage{0.24\textwidth}
	\includegraphics[width=\linewidth]{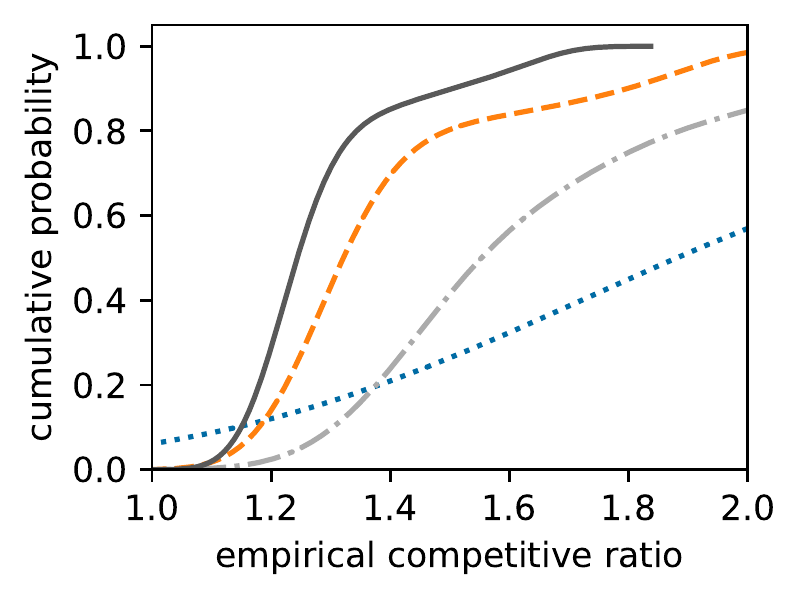}\vspace{-1.2em}
    \caption*{(d) CDF} 
	\endminipage\hfill
    \caption{Experiments on Ontario, Canada carbon trace, with $\theta = 12.0\bar{6}$, and $T=48$.\\  (a): Changing job length $k$ w.r.t. time horizon $T$ ($x$-axis), vs. competitive ratio  (b): Changing switching cost $\beta$ w.r.t. $U$ ($x$-axis), vs. competitive ratio  (c): Different volatility levels w.r.t. $U$ ($x$-axis), vs. competitive ratio\\ (d): Cumulative distribution function of competitive ratios}
    \label{fig:caTrace}
    \end{center}
\end{figure*}

In Figure \ref{fig:slack}, we show results for three different values of time horizon $T$ in each carbon trace, with fixed $\beta$, fixed $k = \lceil T/6 \rceil$, and no added volatility.  Although our experiments test three distinct values for $T$, we later observe that the \textit{ratio between $k$ and $T$} is the primary factor which changes the observed performance of the algorithms we test; in this figure, \DTPR and the benchmark algorithms compare very similarly on the same carbon trace for different $T$ values.  As such, we set $T = 48$ in the rest of the experiments in this section for brevity.  This represents a \textit{slack value} of $48$ hours.

In the first experiment, we test all algorithms for different job lengths $k$ in the range from $4$ hours to $T/2$ ($24$ hours).  The switching cost $\beta$ is non-zero and fixed, and no volatility is added to the carbon trace.  By testing different values for $k$, this experiment tests different ratios between the workload length and the slack provided to the algorithm.  In Figures \ref{fig:caTrace}(a), \ref{fig:usTrace}(a), and \ref{fig:nzTrace}(a), we show that the observed competitive ratio of \DTPRmin outperforms the benchmark algorithms, and it compares particularly favorably for \textit{short} job lengths.  Averaging over all regions and job lengths, the competitive ratio achieved by \DTPRmin is a $11.4$\% improvement on the carbon-agnostic method, a $14.0$\% improvement on the $k$-min search algorithm, and a $5.5$\% improvement on the constant threshold algorithm.

\begin{figure*}[t]
    \begin{center}
    \minipage{0.7\textwidth}
    \includegraphics[width=\linewidth]{img/legend.png}
    \endminipage\hfill\\
	\minipage{0.24\textwidth}
	\includegraphics[width=\linewidth]{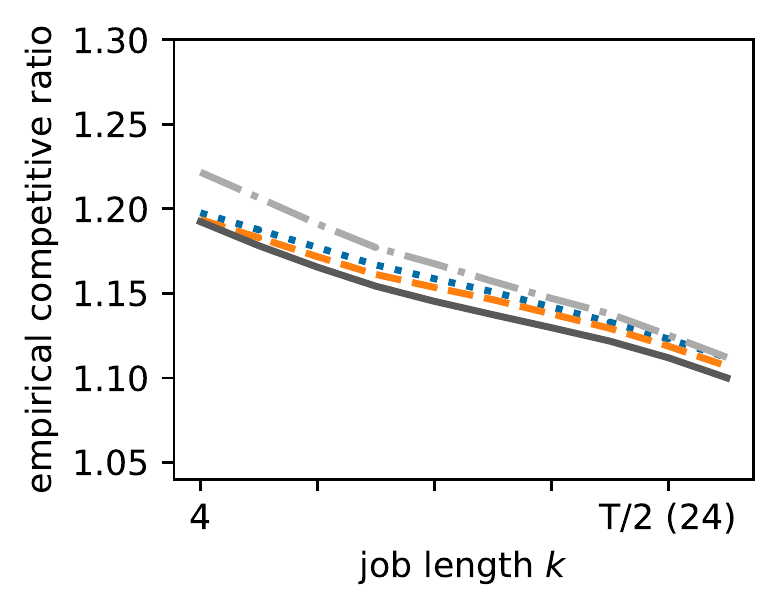}\vspace{-1.2em}
    \caption*{(a) Changing $k$}
	\endminipage\hfill
	\minipage{0.24\textwidth}
	\includegraphics[width=\linewidth]{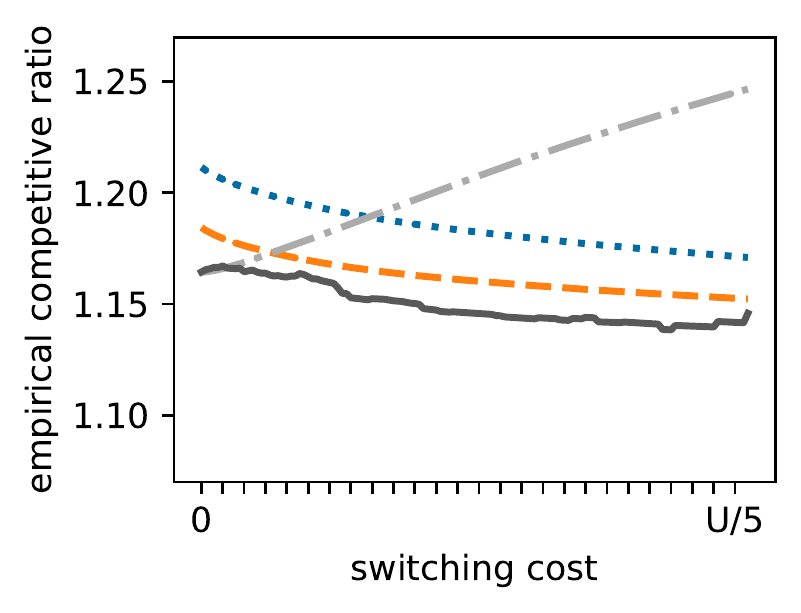}\vspace{-1.2em}
    \caption*{(b) Changing $\beta$}
	\endminipage\hfill
    \minipage{0.24\textwidth}
	\includegraphics[width=\linewidth]{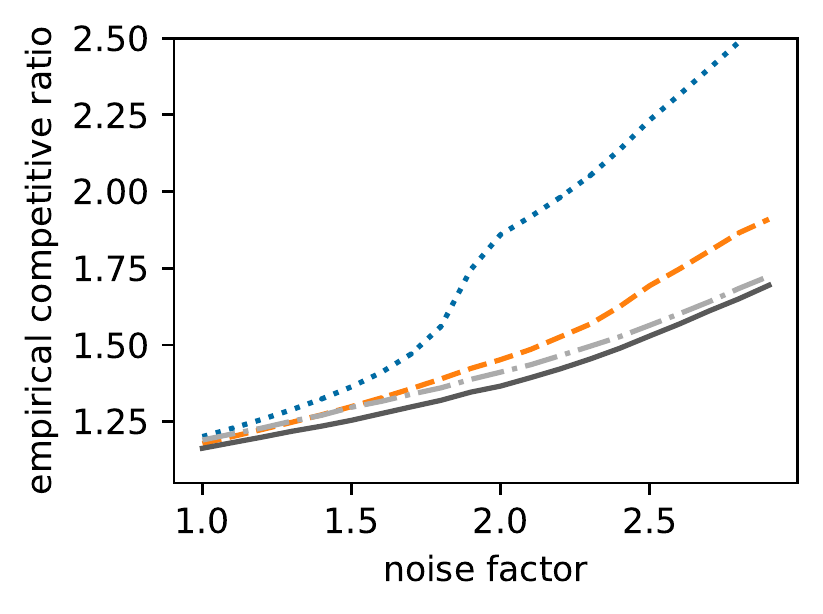}\vspace{-1.2em}
    \caption*{(c) Changing volatility}
	\endminipage\hfill
    \minipage{0.24\textwidth}
	\includegraphics[width=\linewidth]{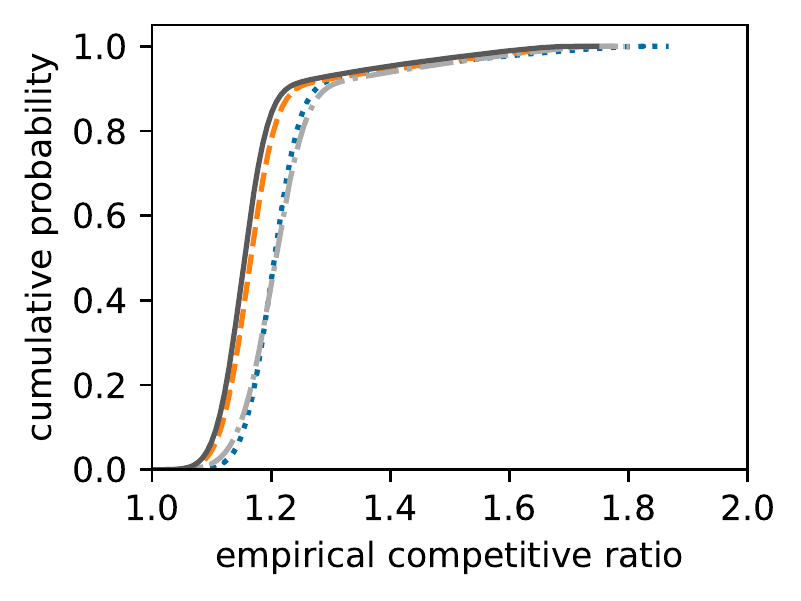}\vspace{-1.2em}
    \caption*{(d) CDF} %
	\endminipage\hfill
    \caption{Experiments on U.S. Pacific Northwest carbon trace, with $\theta = 36$, and $T=48$.\\  (a): Changing job length $k$ w.r.t. time horizon $T$ ($x$-axis), vs. competitive ratio  (b): Changing switching cost $\beta$ w.r.t. $U$ ($x$-axis), vs. competitive ratio  (c): Different volatility levels w.r.t. $U$ ($x$-axis), vs. competitive ratio\\ (d): Cumulative distribution function of competitive ratios}
    \label{fig:usTrace}
    \end{center}
\end{figure*}

In the second experiment, we test all algorithms for different switching costs $\beta$ in the range from $0$ to $U/5$.  The job length $k$ is set to $10$ hours, and no volatility is added to the carbon trace.  By testing different values for $\beta$, this experiment tests how an increasing switching cost impacts the performance of \DTPRmin with respect to other algorithms which do not explicitly consider the switching cost.  In Figures \ref{fig:caTrace}(b), \ref{fig:usTrace}(b), and \ref{fig:nzTrace}(b), we show that the observed competitive ratio of \DTPRmin outperforms the benchmark algorithms for most values of $\beta$ in all regions.  Unsurprisingly, the carbon-agnostic technique (which incurs minimal switching cost) performs better as $\beta$ grows.  While the constant threshold algorithm has relatively consistent performance, the $k$-min search algorithm performs noticeably worse as $\beta$ grows.  Averaging over all regions and switching cost values, the competitive ratio achieved by \DTPRmin is a $18.2$\% improvement on the carbon-agnostic method, a $8.9$\% improvement on the $k$-min search algorithm, and a $4.1$\% improvement on the constant threshold algorithm.

\begin{figure*}[t]
    \begin{center}
    \minipage{0.7\textwidth}
    \includegraphics[width=\linewidth]{img/legend.png}
    \endminipage\hfill\\
	\minipage{0.24\textwidth}
	\includegraphics[width=\linewidth]{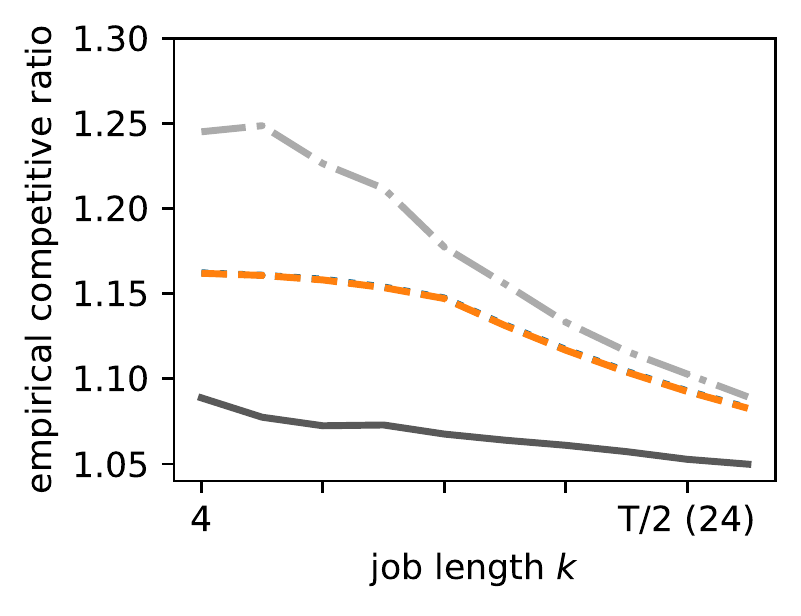}\vspace{-1.2em}
    \caption*{(a) Changing $k$}
	\endminipage\hfill
	\minipage{0.24\textwidth}
	\includegraphics[width=\linewidth]{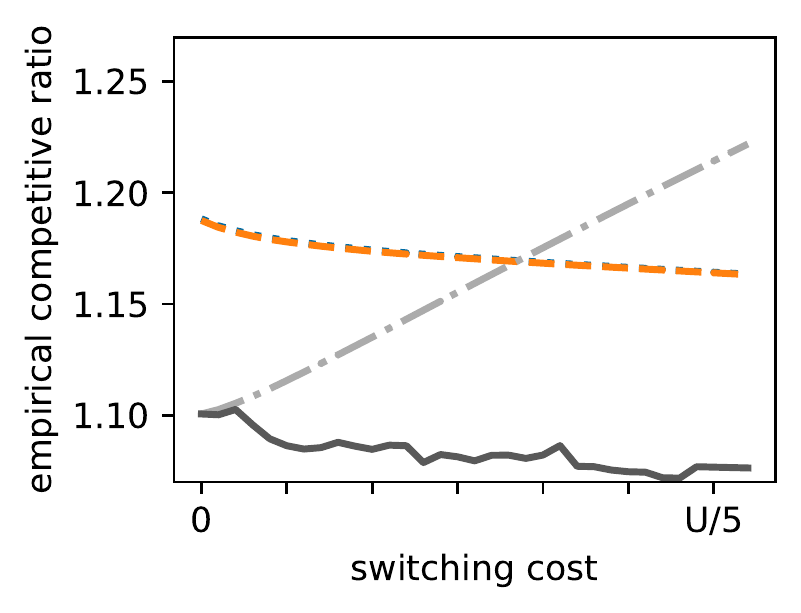}\vspace{-1.2em}
    \caption*{(b) Changing $\beta$}
	\endminipage\hfill
    \minipage{0.24\textwidth}
	\includegraphics[width=\linewidth]{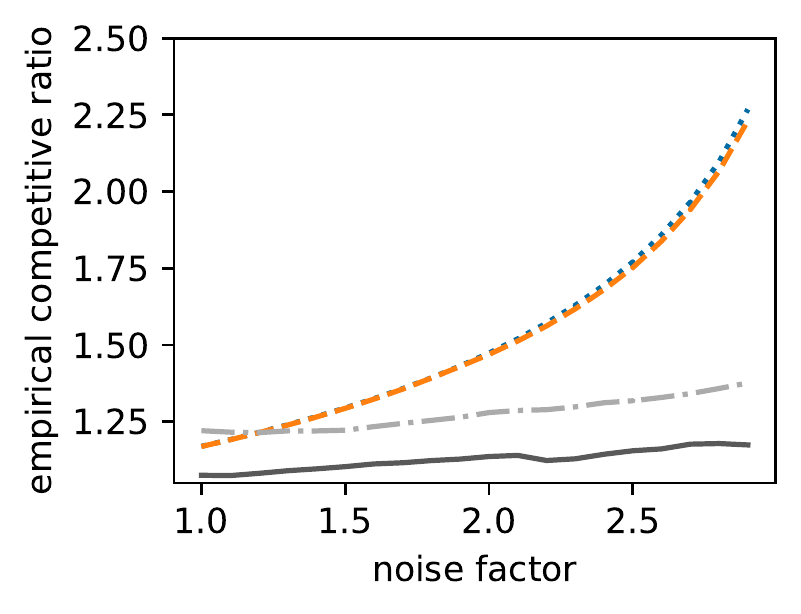}\vspace{-1.2em}
    \caption*{(c) Changing volatility}
	\endminipage\hfill
    \minipage{0.24\textwidth}
	\includegraphics[width=\linewidth]{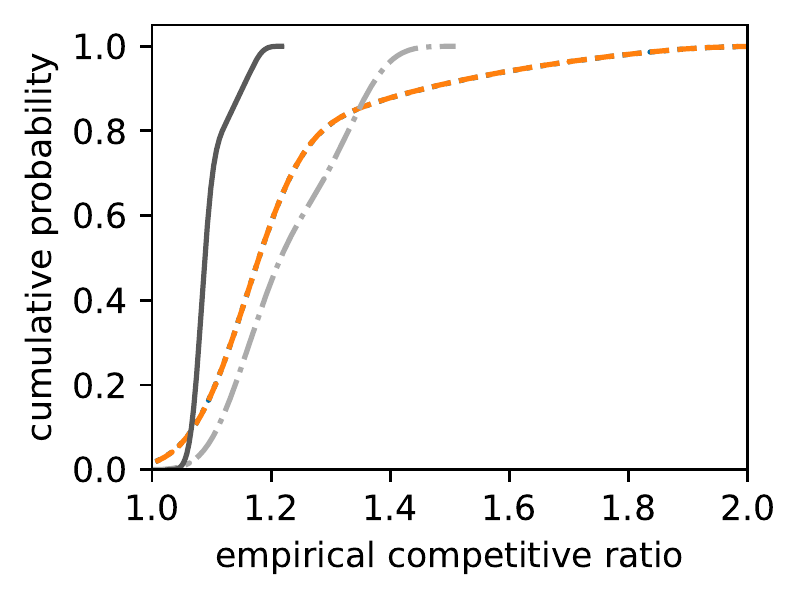}\vspace{-1.2em}
    \caption*{(d) CDF} %
	\endminipage\hfill
    \caption{Experiments on New Zealand carbon trace, with $\theta = 3.0\bar{5}$, and $T=48$.\\  (a): Changing job length $k$ w.r.t. time horizon $T$ ($x$-axis), vs. competitive ratio  (b): Changing switching cost $\beta$ w.r.t. $U$ ($x$-axis), vs. competitive ratio  (c): Different volatility levels w.r.t. $U$ ($x$-axis), vs. competitive ratio\\ (d): Cumulative distribution function of competitive ratios}
    \label{fig:nzTrace}
    \end{center}
\end{figure*}

In the final experiment, we test all algorithms on sequences with different volatility.  The job length $k$ and switching cost $\beta$ are both fixed.  We add volatility by setting a \textit{noise factor} from the range $1.0$ to $3.0$.   By testing different values for this volatility, this experiment tests how each algorithm handles larger fluctuations in the carbon intensity of consecutive time steps.  In Figures \ref{fig:caTrace}(c), \ref{fig:usTrace}(c), and \ref{fig:nzTrace}(c), we show that the observed competitive ratio of \DTPRmin outperforms the benchmark algorithms for all noise factors in all regions.  Intuitively, higher volatility values cause the online algorithms to perform worse in general.  Averaging over all regions and noise factors, the competitive ratio achieved by \DTPRmin is a $53.6$\% improvement on the carbon-agnostic method, a $13.5$\% improvement on the $k$-min search algorithm, and a $14.3$\% improvement on the constant threshold algorithm.

By averaging over all experiments for a given region, we obtain the cumulative distribution function plot for each algorithm's competitive ratio in Figures \ref{fig:caTrace}(d), \ref{fig:usTrace}(d), and \ref{fig:nzTrace}(d).  Compared to the carbon-agnostic, constant threshold, and $k$-min search algorithms, \DTPRmin achieves a lower average empirical competitive ratio distribution for all tested regions.  Across \textit{all regions} at the 95th percentile, \DTPRmin achieves a worst-case empirical competitive ratio of $1.40$.  This represents a $48.2$\% improvement over the \textit{carbon-agnostic} algorithm, and improvements of $15.6$\% and $14.4$\% over the $k$-min search and constant threshold \textit{switching-cost-agnostic} algorithms, respectively.

%% file: relwork.tex
This paper contributes to three lines of work: (i) work on online search and related problems, e.g., $k$-search, one-way trading, and online knapsack; (ii) work on online optimization problems with switching costs, e.g., metrical task systems and convex function chasing; and (iii) work on carbon-aware load shifting.  We describe the relationship to each below.

\paragraph{Online Search.}
The \OPR problem is closely related to the online $k$-search problem~\cite{Lee:22,Lorenz:08}, as discussed in the introduction and Section~\ref{sec:OTA}. It also has several similar counterparts, including online conversion problems such as one-way trading~\cite{ElYaniv:01,mohr2014online,Sun:21,Damaschke:07} and online knapsack problems~\cite{Zhou:08,sun2022online,Yang2021Competitive}, with practical applications to stock trading~\cite{Lorenz:08}, cloud pricing~\cite{zhang2017optimal}, electric vehicle charging~\cite{sun2020competitive}, etc. The $k$-search problem can be viewed as an integral version of the online conversion problem, while the general online conversion problem allows continuous one-way trading. The basic online knapsack problem studies how to pack arriving items of different sizes and values into a knapsack with limited capacity, while its extensions to item departures~\cite{zhang2017optimal,sun2022online} and multidimensional capacity~\cite{Yang2021Competitive} have also been studied recently. 
Another line of research leverages machine learning predictions of unknown future inputs to design learning-augmented online algorithms for online $k$-search~\cite{Lee:22} and online conversion~\cite{sun2022online}. However, to the best of our knowledge, none of these works consider the switching cost of changing decisions.  Thus, this work is the first to incorporate switching costs to the $k$-search framework.

\paragraph{Metrical Task Systems.}

The metrical task systems (\texttt{MTS}) problem was introduced by Borodin et al. in \cite{Borodin:92}. 
Several decades of progress on upper and lower bounds on the competitive ratio of \texttt{MTS} recently culminated with a tight bound of $\Theta(\log^2 n)$ for the competitive ratio of \texttt{MTS} on an arbitrary $n$-point metric space, with $\Theta(\log n)$ being possible on certain metric spaces such as trees \cite{bubeckMetricalTaskSystems2021, bubeckRandomizedServerConjecture2022}. Several modified forms of \texttt{MTS} have also seen significant attention in the literature, such as smoothed online convex optimization (\texttt{SOCO}) and convex function chasing (\texttt{CFC}), in which the decision space is an $n$-dimensional normed vector space and cost functions are restricted to be convex \cite{FriedmanLinial:93, Lin:12}. The best known upper and lower bounds on the competitive ratio of \texttt{CFC} are $O(n)$ and $\Omega(\sqrt{n})$, respectively, in $n$-dimensional Euclidean spaces \cite{bubeckChasingNestedConvex2019, sellkeChasingConvexBodies2020}. However, algorithms with competitive ratios independent of dimension can be obtained for certain special classes of functions, such as $\alpha$-polyhedral functions \cite{chenSmoothedOnlineConvex2018a}. A number of recent works have also investigated the design of learning-augmented algorithms for various cases of \texttt{CFC}/\texttt{SOCO} and \texttt{MTS} which exploit the performance of machine-learned predictions of the optimal decisions \cite{antoniadisOnlineMetricAlgorithms2020a, Christianson:22, christianson2023mts, liExpertCalibratedLearningOnline2022, ruttenSmoothedOnlineOptimization2022}. The key characteristic distinguishing \OPR from \texttt{MTS} and its variants is the presence of a terminal deadline constraint. None of the algorithms for \texttt{MTS}-like problems are designed to handle such long-term constraints while maintaining any sort of competitive guarantee.

\paragraph{Carbon-Aware Temporal Workload Shifting.}  The goal of shifting workloads in time to allow more sustainable operations of data centers has been of interest for more than a decade, e.g.,~\cite{gupta2019combining,liu2011greening,liu2012renewable,lin2012dynamic}.  Traditionally, such papers have used models that build on one of convex function chasing, $k$-search, or online knapsack to design algorithms; however such models do not capture both the switching costs and long-term deadlines that are crucial to practical deployment.  
In recent years, the load shifting literature has focused specifically on reducing the carbon footprint of operations, e.g.,~\cite{radovanovic2022carbon,acun2022holistic,bashir2021enabling,Wiesner:21}.  Perhaps most related to this paper is~\cite{Wiesner:21}, which explores the problem of carbon-aware temporal workload shifting and proposes a threshold-based algorithm that suspends the job when the carbon intensity is higher than a threshold value and resumes it when it drops below the threshold. However, it does not consider switching nor does it provide any deadline guarantees. 
Other recent work on carbon-aware temporal shifting seeks to address the resultant increase in job completion times. 
In ~\cite{Souza:23}, authors leverage the pause and resume approach to reduce the carbon footprint of ML training and high-performance computing applications such as BLAST~\cite{blast}. 
However, instead of resuming at normal speed ($1\times$) during the low carbon intensity periods, their applications resume operation at a faster speed ($m\times$), where the scale factor $m$ depends on the application characteristics.
It uses a threshold-based approach to determine the low carbon intensity periods but does not consider switching costs or provide any deadline guarantees. An interesting future direction is to extend the \DTPR algorithms to consider the ability to scale up speed after resuming jobs.

%% file: conclusion.tex
Motivated by carbon-aware load shifting, we introduce and study the online pause and resume problem (\OPR), which bridges gaps between several related problems in online optimization.  To our knowledge, it is the first online optimization problem that includes both long-term constraints and switching costs. Our main results provide optimal online algorithms for the minimization and maximization variants of this problem, as well as lower bounds for the competitive ratio of any deterministic online algorithm.  Notably, our proposed algorithms match existing optimal results for the related $k$-search problem when the switching cost is $0$, and improve on the $k$-min search competitive bounds for non-zero switching cost. The key to our results is a novel double threshold algorithm that we expect to be applicable in other online problems with switching costs. 

There are a number of interesting directions in which to continue the study of \OPR.  We have highlighted the application of \OPR to carbon-aware load shifting, but \OPR also applies to many other problems where pricing changes over time and frequent switching is undesirable.  Pursuing these applications is important. Theoretically, there are several interesting open questions.  First, considering the target application of carbon-aware load shifting, some workloads are \textit{highly parallelizable}~\cite{Souza:23}, which adds another dimension of scaling to the problem (i.e., instead of choosing to run 1 unit of the job in each time slot, the online player must decide how many units to allocate at each time slot).  This makes the theoretical problem more challenging, and is an important consideration for future work.  Additionally, very recent work has incorporated machine-learned advice to achieve better performance on related online problems, including $k$-search \cite{Lee:22,Sun:21}, \texttt{CFC/SOCO} \cite{Christianson:22, liExpertCalibratedLearningOnline2022}, and \texttt{MTS} \cite{antoniadisOnlineMetricAlgorithms2020a, christianson2023mts, ruttenSmoothedOnlineOptimization2022}. Designing learning-augmented algorithms for \OPR is a very promising line of future work, particularly considering applications such as carbon-aware load shifting, where predictions can significantly improve the algorithm's understanding of the future in the best case.

%% file: appendix.tex
\section{Case Study Results for \DTPRmax Algorithm}
\label{appendix:maxExp}

\begin{figure*}[t]
    \minipage{\textwidth}
    \includegraphics[width=\linewidth]{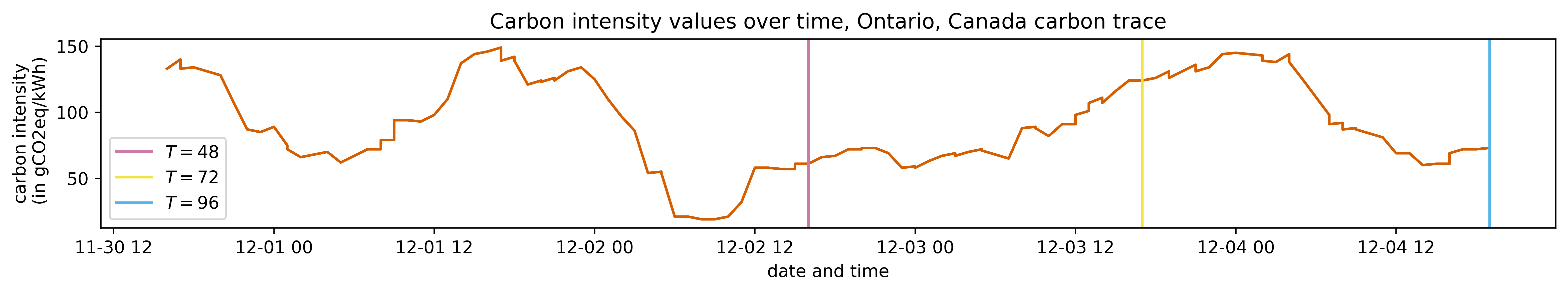}
    \endminipage\hfill\\
	\minipage{\textwidth}
    \includegraphics[width=\linewidth]{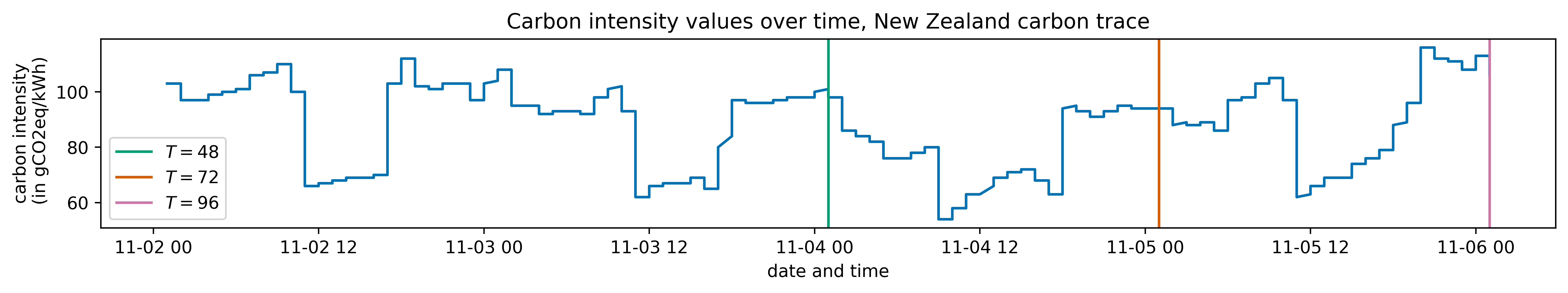}
    \endminipage\hfill\\
    \minipage{\textwidth}
    \includegraphics[width=\linewidth]{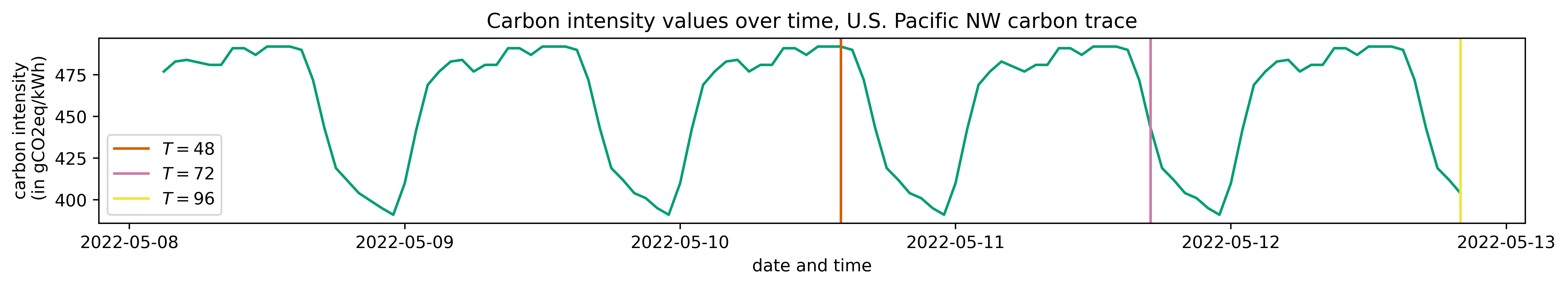}
    \endminipage\hfill
    \vspace{-1em}
    \caption{Carbon intensity (in gCO$_2$eq/kWh) values plotted for each region tested in our numerical experiments, with one-hour granularity.  We plot a representative random interval of $96$ hours, with vertical lines demarcating the different values for $T$ (time horizon) tested in our experiments.  In all regions, carbon values roughly follow a diurnal (daily cycle) pattern.  Actual values and observed intensities significantly vary in different regions.}
    \label{fig:traceVis}
\end{figure*}

\begin{algorithm}
\caption{Double Threshold Pause and Resume for \OPRmax (\DTPRmax)}\label{alg:max}
\begin{algorithmic}[1]
\Require threshold values $\{u_i\}_{i \in [1, k]}$ and $\{\ell_i\}_{i \in [1, k]}$ defined in Equation~\eqref{eq:maxthres}, deadline $T$
\Ensure online decisions $\{x_t\}_{t \in [1, T]}$
\State \textbf{initialize: } i = 1;
\While{price $c_t$ arrives \textbf{and} $i \leq k$}
\If{$(k - i) \geq (T - t)$} \Comment{\footnotesize\texttt{close to the deadline $T$, we must accept remaining prices}\normalsize}
    \State price $c_t$ is accepted, set $x_t = 1$ \label{line:max-force}
\ElsIf{$x_{t-1} = 0$} \Comment{\footnotesize\texttt{If previous price was not accepted}\normalsize}
    \If{$c_t \geq u_i$} \; price $c_t$ is accepted, set $x_t = 1$
    \Else \; price $c_t$ is rejected, set $x_t = 0$
    \EndIf
\ElsIf{$x_{t-1} = 1$} \Comment{\footnotesize\texttt{If previous price was accepted}\normalsize}
    \If{$c_t \geq \ell_i$}\; price $c_t$ is accepted, set $x_t = 1$
    \Else\; price $c_t$ is rejected, set $x_t = 0$
    \EndIf
\EndIf
\State \textnormal{update} $i = i + x_t$
\EndWhile
\end{algorithmic}
\end{algorithm}

This section presents and discusses the deferred experimental results for the \DTPRmax algorithm (pseudocode summarized in Algorithm \ref{alg:max}) in the carbon-aware temporal workload shifting case study.  We evaluate \DTPRmax against the same benchmark algorithms described in Section \ref{sec:expsetup}.

For the \textit{maximization metric}, we consider the \textit{percentage of carbon-free electricity} powering the grid.  At each time step $t$, the electricity supply has a carbon-free percentage $c_t$, i.e., if the job is being processed during time slot $t$ ($x_t = 1$), the electricity powering the data center's is $c_t\%$ \textit{carbon-free}, and the objective is to maximize this percentage over all $k$ slots of the active running of the workload.

In these maximization experiments, the switching-cost-agnostic $k$-max-search algorithm chooses to run the $i$th hour of the job during the first time slot where the carbon-free supply is at least $\Phi_i$.  Similarly, the constant threshold algorithm chooses to run the job whenever the carbon-free supply is at least $\sqrt{UL}$.  We set $L$ and $U$ to be the minimum and maximum carbon-free supply percentages over the entire trace being studied.

\begin{figure*}[t]
    \begin{center}
    \minipage{0.7\textwidth}
    \includegraphics[width=\linewidth]{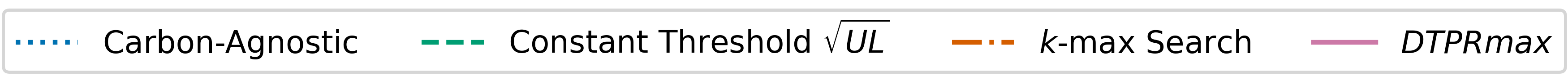}
    \endminipage\hfill\\
	\minipage{0.24\textwidth}
	\includegraphics[width=\linewidth]{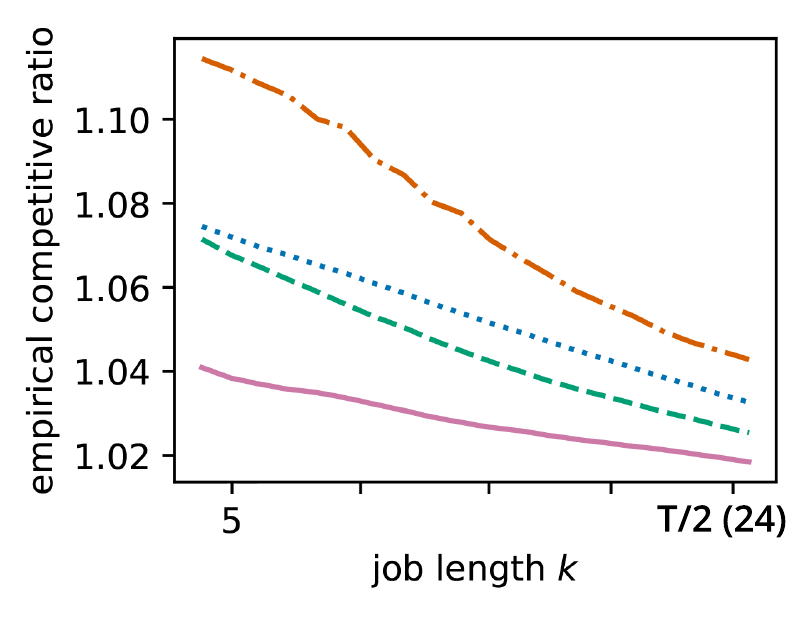}\vspace{-1em}
    \caption*{(a) Changing $k$}
	\endminipage\hfill
	\minipage{0.24\textwidth}
	\includegraphics[width=\linewidth]{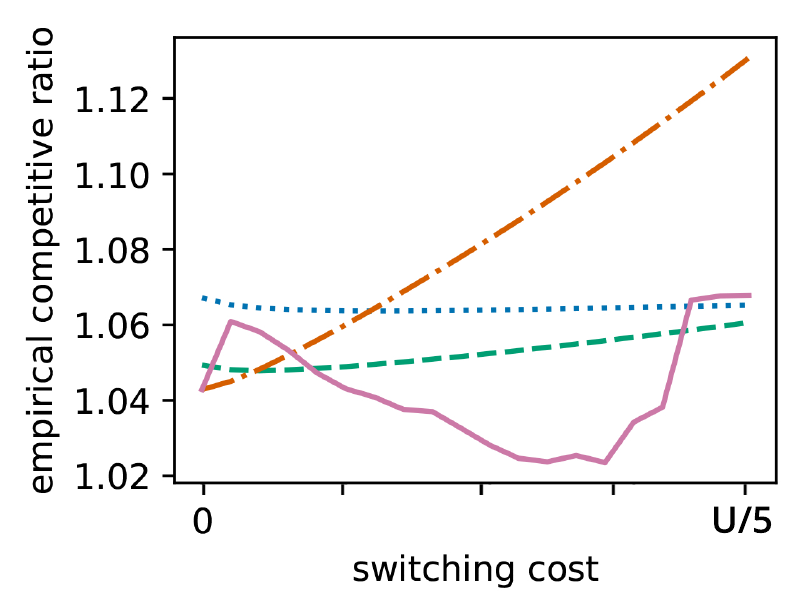}\vspace{-1em}
    \caption*{(b) Changing $\beta$}
	\endminipage\hfill
    \minipage{0.24\textwidth}
	\includegraphics[width=\linewidth]{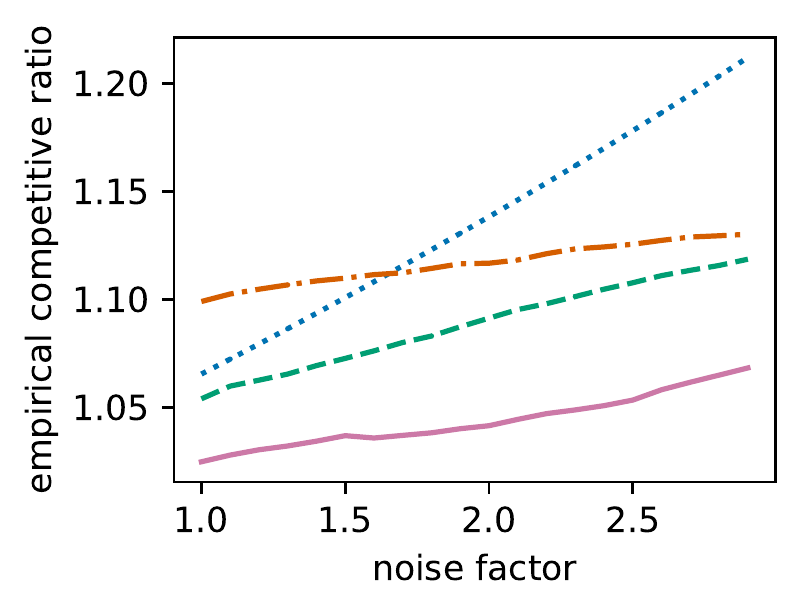}\vspace{-1em}
    \caption*{(c) Changing volatility}
	\endminipage\hfill
    \minipage{0.24\textwidth}
	\includegraphics[width=\linewidth]{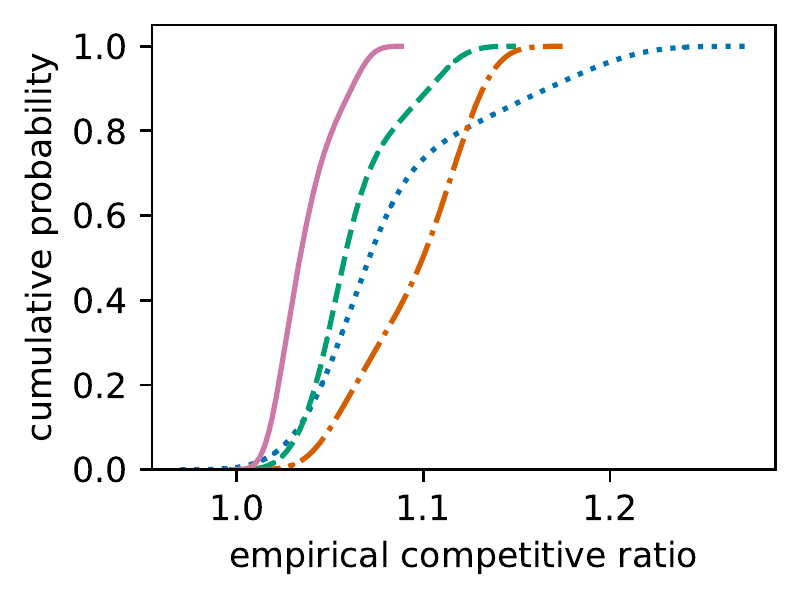}\vspace{-1em}
    \caption*{(d) CDF} %
	\endminipage\hfill
    \caption{Maximization experiments on Ontario, Canada carbon trace, with $\theta \approx 1.51$ and $T=48$.\\  (a): Changing job length $k$ w.r.t. time horizon $T$ ($x$-axis), vs. competitive ratio  (b): Changing switching cost $\beta$ w.r.t. $U$ ($x$-axis), vs. competitive ratio  (c): Different volatility levels w.r.t. $U$ ($x$-axis), vs. competitive ratio (d): Cumulative distribution function of competitive ratios}\label{fig:caTraceMAX}
    \end{center}
\end{figure*}

As in Section \ref{sec:expresults}, our focus is on the competitive ratio (lower competitive ratio is better).  We report the performance of all algorithms for each experiment setting, in each tested region.

In the first experiment, we test all algorithms for different job lengths $k$ in the range from $4$ hours to $T/2 (24)$.  The switching cost $\beta$ is non-zero and fixed, and no volatility is added to the carbon trace.  By testing different values for $k$, this experiment tests different ratios between the workload length and the slack provided to the algorithm.  In Figures \ref{fig:caTraceMAX}(a), \ref{fig:usTraceMAX}(a), and \ref{fig:nzTraceMAX}(a), we show that the observed average competitive ratio of \DTPRmax narrowly outperforms the benchmark algorithms for all values of $k$ in all regions, and it compares particularly favorably for \textit{short} job lengths.  Averaging over all regions and job lengths, the competitive ratio achieved by \DTPRmax is a $4.9$\% improvement on the carbon-agnostic method, a $8.4$\% improvement on the $k$-max search algorithm, and a $2.1$\% improvement on the constant threshold algorithm.%

In the second experiment, we test all algorithms for different switching costs $\beta$ in the range from $0$ to $U/5$.  The job length $k$ is set to $10$ hours, and no volatility is added to the carbon trace.  By testing different values for $\beta$, this experiment tests how an increasing switching cost impacts the performance of \DTPRmax with respect to other algorithms which do not explicitly consider the switching cost.  In Figures \ref{fig:caTraceMAX}(b), \ref{fig:usTraceMAX}(b), and \ref{fig:nzTraceMAX}(b), we show that the average competitive ratio of \DTPRmax notably outperforms the other algorithms for a wide range of $\beta$ values in all regions.  Unsurprisingly, the carbon-agnostic technique (which only incurs a switching cost of $2\beta$) is more competitive as $\beta$ grows.  The $k$-max search algorithm performs noticeably worse as $\beta$ grows. While the constant threshold algorithm has relatively consistent performance, the $k$-max search algorithm performs noticeably worse as $\beta$ grows.  Averaging over all regions and switching cost values, the competitive ratio achieved by \DTPRmax is a $2.5$\% improvement on the carbon-agnostic method, a $6.4$\% improvement on the $k$-max search algorithm, and a $0.1$\% improvement on the constant threshold algorithm. %

\begin{figure*}[t]
    \begin{center}
    \minipage{0.7\textwidth}
    \includegraphics[width=\linewidth]{img/max/legend.png}
    \endminipage\hfill\\
	\minipage{0.24\textwidth}
	\includegraphics[width=\linewidth]{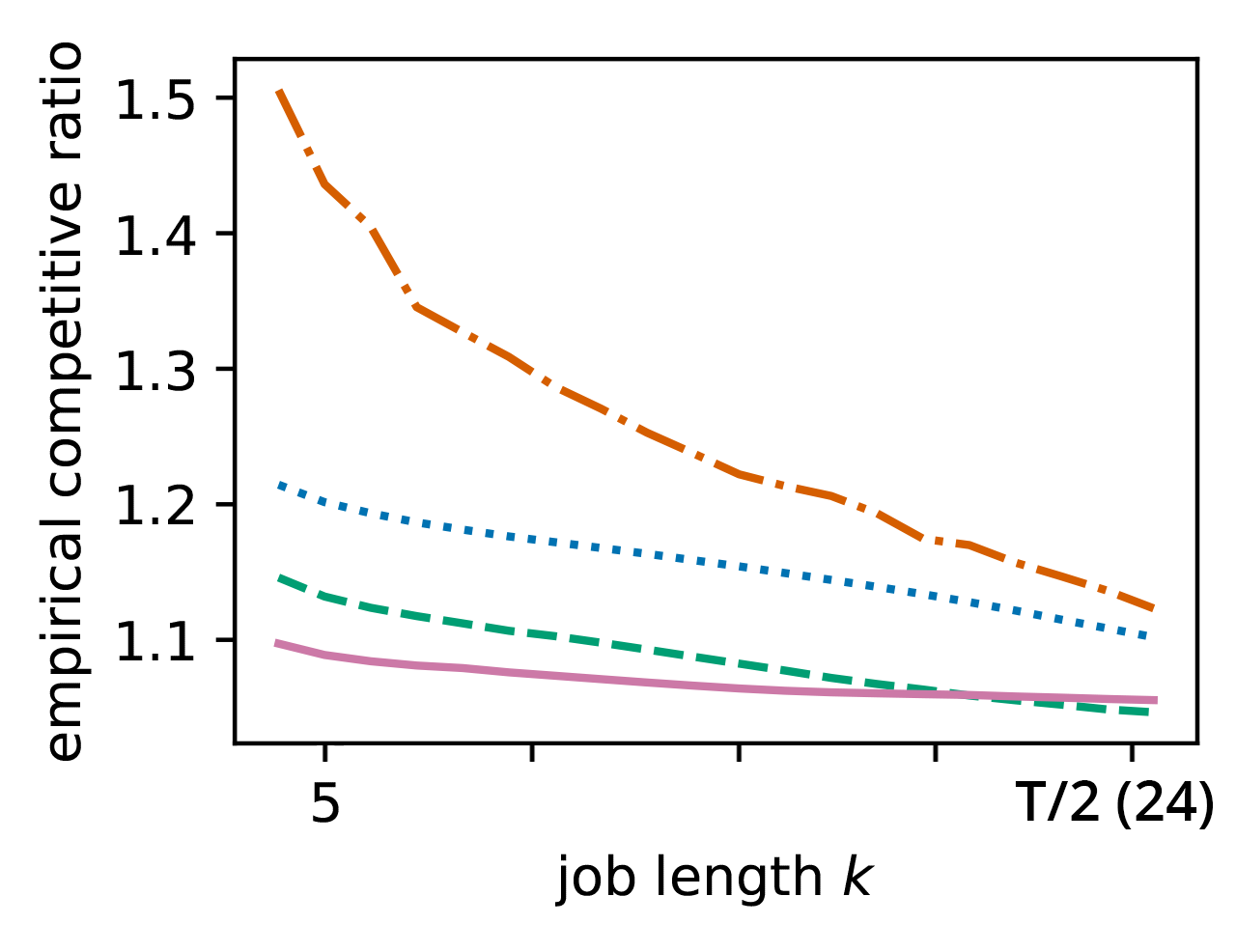}\vspace{-1em}
    \caption*{(a) Changing $k$}
	\endminipage\hfill
	\minipage{0.24\textwidth}
	\includegraphics[width=\linewidth]{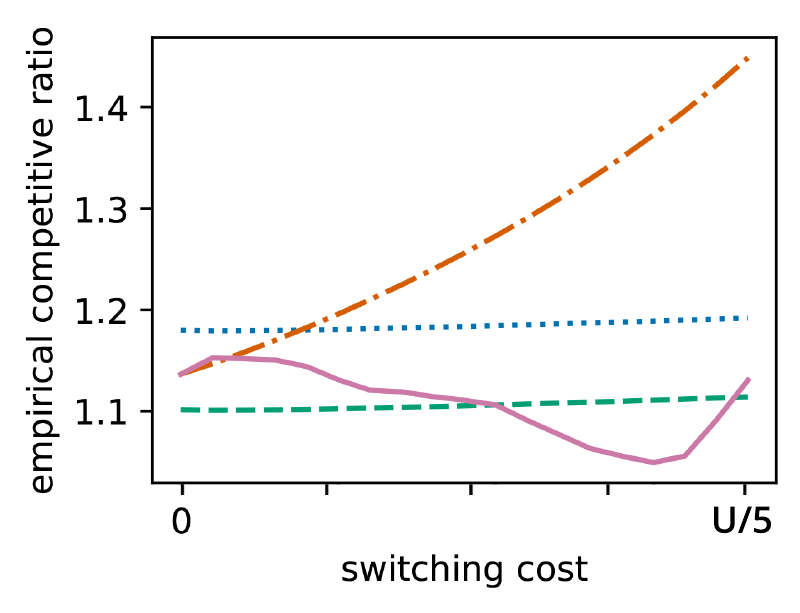}\vspace{-1em}
    \caption*{(b) Changing $\beta$}
	\endminipage\hfill
    \minipage{0.24\textwidth}
	\includegraphics[width=\linewidth]{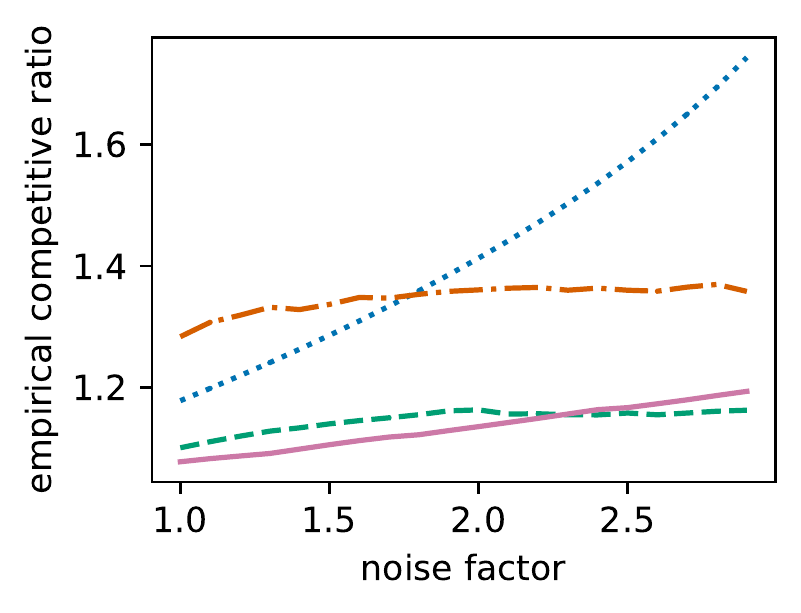}\vspace{-1em}
    \caption*{(c) Changing volatility}
	\endminipage\hfill
    \minipage{0.24\textwidth}
	\includegraphics[width=\linewidth]{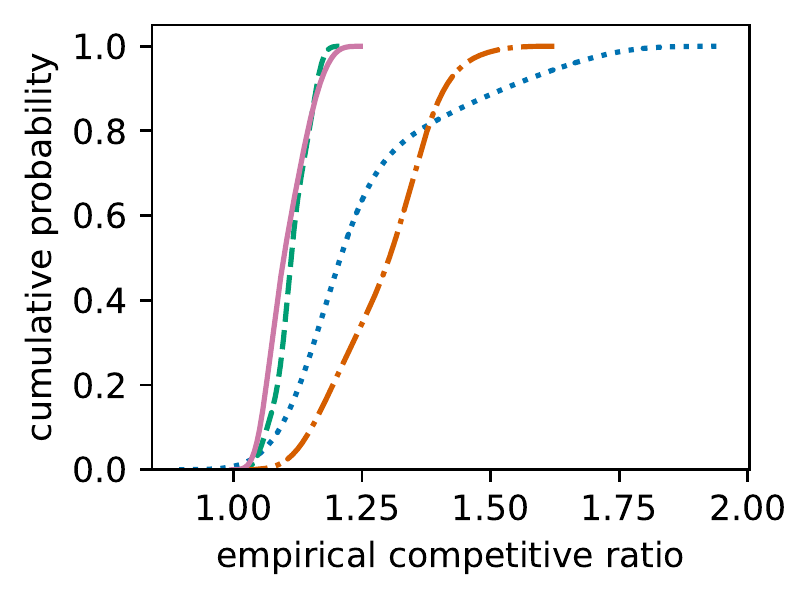}\vspace{-1em}
    \caption*{(d) CDF} %
	\endminipage\hfill
    \caption{Maximization experiments on U.S. Pacific Northwest carbon trace, with $\theta \approx 5.24$ and $T=48$.\\  (a): Changing job length $k$ w.r.t. time horizon $T$ ($x$-axis), vs. competitive ratio  (b): Changing switching cost $\beta$ w.r.t. $U$ ($x$-axis), vs. competitive ratio  (c): Different volatility levels w.r.t. $U$ ($x$-axis), vs. competitive ratio (d): Cumulative distribution function of competitive ratios}\label{fig:usTraceMAX}
    \end{center}
\end{figure*}

\begin{figure*}[t]
    \begin{center}
    \minipage{0.7\textwidth}
    \includegraphics[width=\linewidth]{img/max/legend.png}
    \endminipage\hfill\\
	\minipage{0.24\textwidth}
	\includegraphics[width=\linewidth]{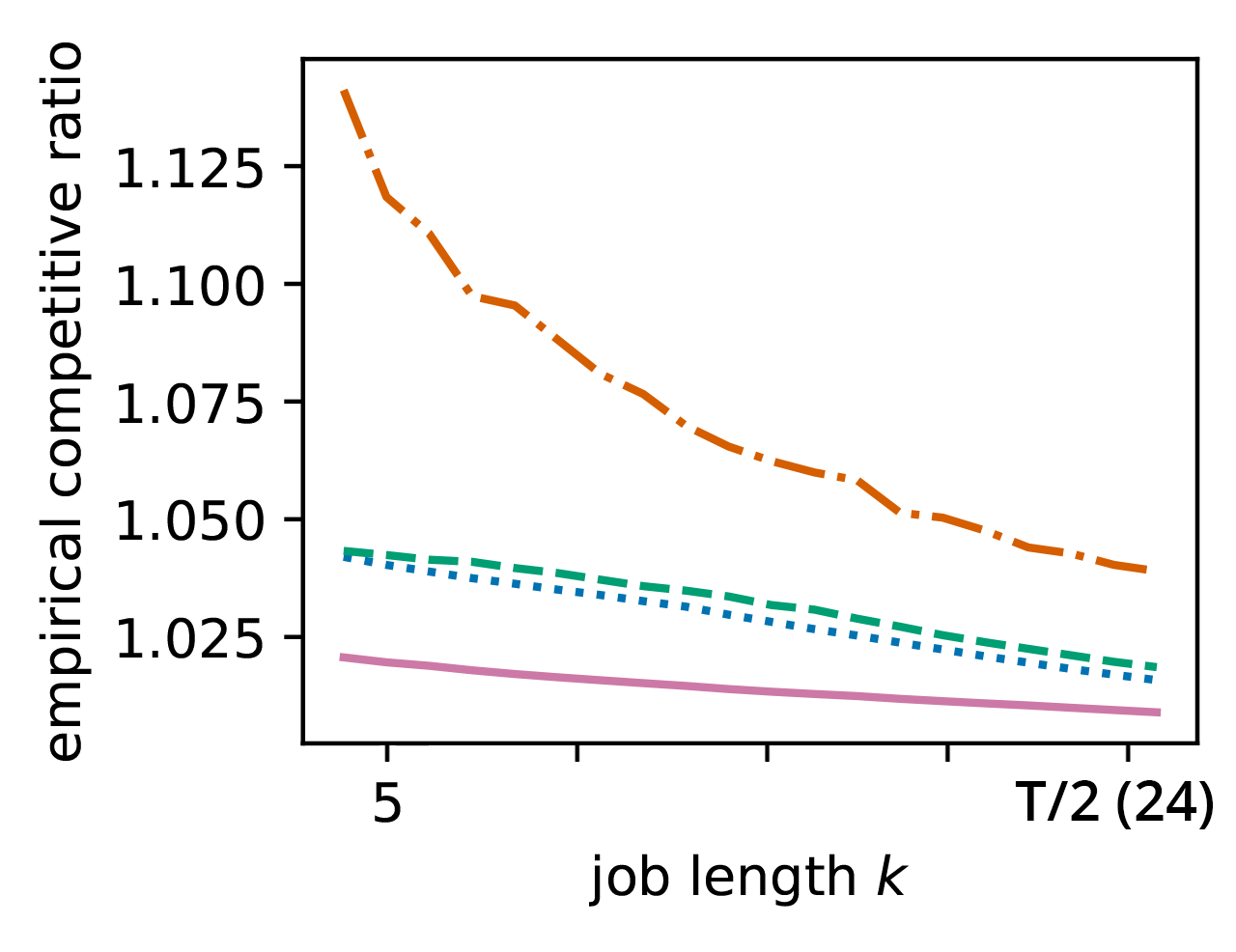}\vspace{-1em}
    \caption*{(a) Changing $k$}
	\endminipage\hfill
	\minipage{0.24\textwidth}
	\includegraphics[width=\linewidth]{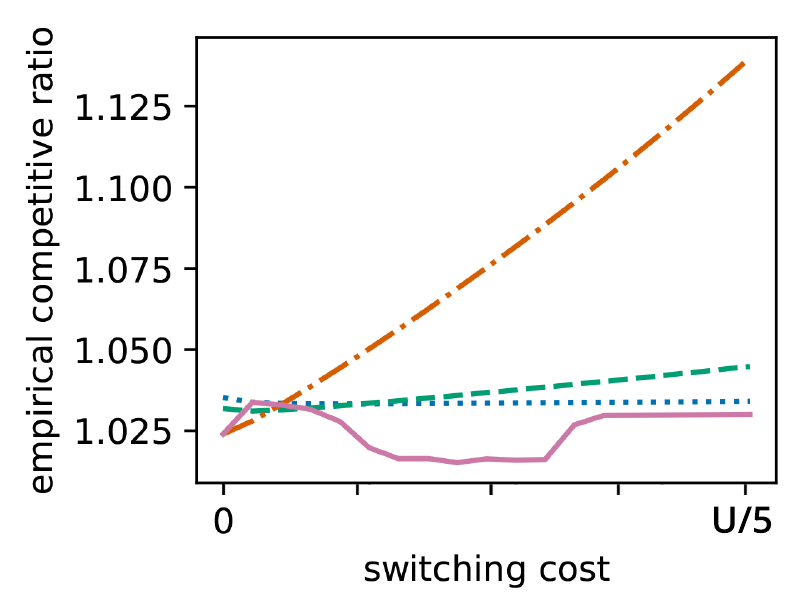}\vspace{-1em}
    \caption*{(b) Changing $\beta$}
	\endminipage\hfill
    \minipage{0.24\textwidth}
	\includegraphics[width=\linewidth]{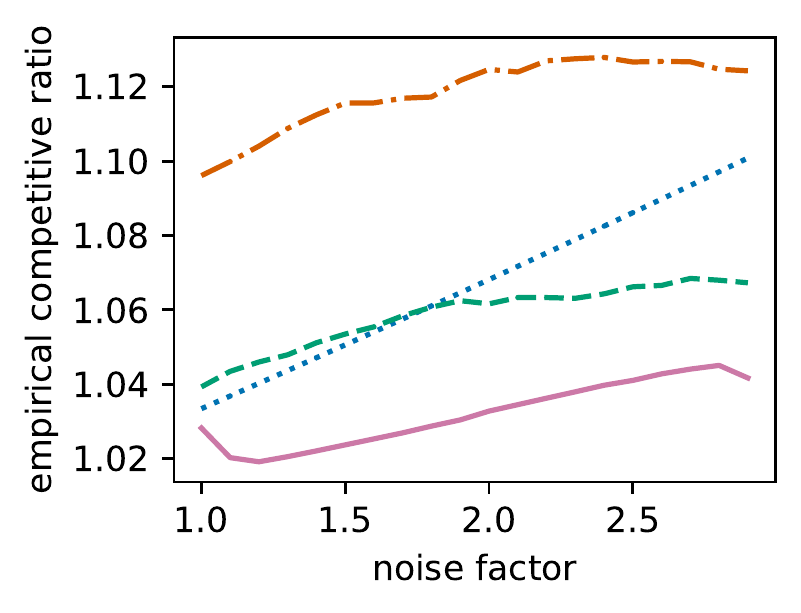}\vspace{-1em}
    \caption*{(c) Changing volatility}
	\endminipage\hfill
    \minipage{0.24\textwidth}
	\includegraphics[width=\linewidth]{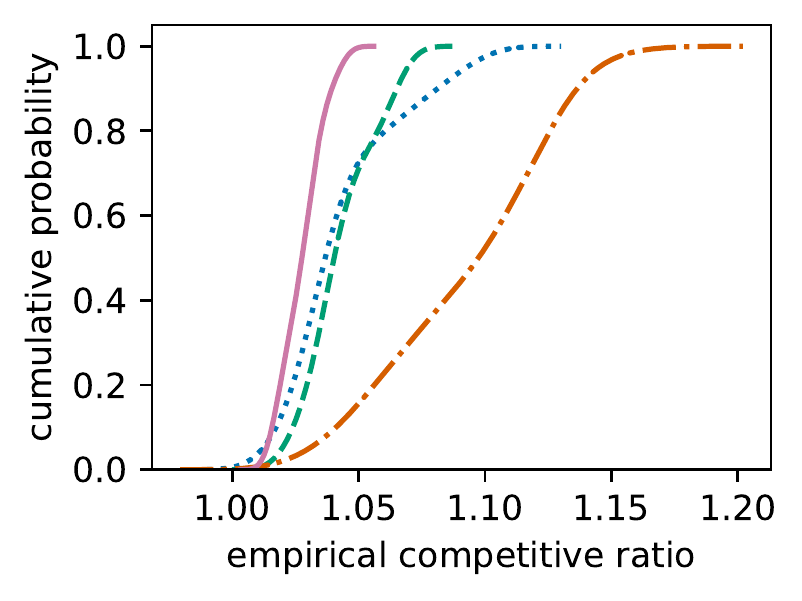}\vspace{-1em}
    \caption*{(d) CDF} %
	\endminipage\hfill
    \caption{Maximization experiments on New Zealand carbon trace, with $\theta \approx 1.35$ and $T=48$.\\  (a): Changing job length $k$ w.r.t. time horizon $T$ ($x$-axis), vs. competitive ratio  (b): Changing switching cost $\beta$ w.r.t. $U$ ($x$-axis), vs. competitive ratio  (c): Different volatility levels w.r.t. $U$ ($x$-axis), vs. competitive ratio (d): Cumulative distribution function of competitive ratios} \label{fig:nzTraceMAX}
    \end{center}
\end{figure*}

In the final experiment, we test all algorithms on sequences with different volatility.  The job length $k$ and switching cost $\beta$ are both fixed.  We add volatility by setting a \textit{noise factor} from the range $1.0$ to $3.0$.   By testing different values for this volatility, this experiment tests how each algorithm handles larger fluctuations in the carbon intensity of consecutive time steps.  In Figures \ref{fig:caTraceMAX}(c), \ref{fig:usTraceMAX}(c), and \ref{fig:nzTraceMAX}(c), we show that the observed average competitive ratio of \DTPRmax outperforms the other algorithms for most noise factors in all regions, with a slight degradation in the Pacific Northwest region.  Intuitively, higher volatility values cause the online algorithms to perform worse in general.  Averaging over all regions and noise factors, the competitive ratio achieved by \DTPRmax is a $13.0$\% improvement on the carbon-agnostic method, a $11.2$\% improvement on the $k$-max search algorithm, and a $2.1$\% improvement on the constant threshold algorithm.

By averaging over all experiments for a given region, we obtain the cumulative distribution function plot for each algorithm's competitive ratio in Figures \ref{fig:caTraceMAX}(d), \ref{fig:usTraceMAX}(d), and \ref{fig:nzTraceMAX}(d).  Compared to the carbon-agnostic, constant threshold, and $k$-max search algorithms, \DTPRmax generally exhibits a lower average empirical competitive ratio over the tested regions.  Notably, all of the algorithms are nearly 1-competitive in our experiments.  Compared to our minimization experiments, \DTPRmax outperforms the baseline algorithms by a smaller margin.  Across \textit{all regions} at the 95th percentile, \DTPRmax achieves a worst-case empirical competitive ratio of $1.08$.  This represents a $16.1$\% improvement over the \textit{carbon-agnostic} algorithm, and improvements of $11.4$\% and $2.19$\% over the $k$-max search and constant threshold \textit{switching-cost-agnostic} algorithms, respectively.

We conjecture that one dynamic contributing to this is the relatively low values of $\theta$ observed for the carbon-free supply percentage in these real-world carbon traces.

\section{Competitive Analysis of \DTPRmax: Proof of Theorem \ref{thm:compPRmax}} \label{app:compPRmax}

Here we prove the \DTPRmax results presented in Theorem~\ref{thm:compPRmax} and Corollary~\ref{cor:max}.
\begin{proof}[Proof of Theorem \ref{thm:compPRmax}] 
For $0 \leq j \leq k$, let $\mathcal{S}_j \subseteq \mathcal{S}$ be the sets of \OPRmax price sequences for which \DTPRmax accepts exactly $j$ prices (excluding the $k-j$ prices it is forced to accept at the end of the sequence).  Then all of the possible price sequences for \OPRmax are represented by $\mathcal{S} = \bigcup_{j=0}^k \mathcal{S}_j$.  By definition, $u_{k+1} = U$.  Let $\epsilon > 0$ be fixed, and define the following two price sequences $\sigma_j$ and $\rho_j$:

$$\forall 0 \leq j \leq k : \sigma_j = u_1, \ell_2, \dots, \ell_j, L, \underbrace{u_{j+1} - \epsilon, \dots, u_{j+1} - \epsilon}_{k}, \underbrace{L, L, \dots, L}_{k}.$$
$$\forall 0 \leq j \leq k : \rho_j = u_1, L, u_2, L, \dots, L, u_j, L, \underbrace{u_{j+1} - \epsilon, \dots, u_{j+1} - \epsilon}_{k}, \underbrace{L, L, \dots, L}_{k}.$$
We have two special cases for $j=0$ and $j=1$.  For $j=0$, we have that $\sigma_0 = \rho_0$, and this sequence simply consists of $u_{1} - \epsilon$ repeated $k$ times, followed by $L$ repeated $k$ times.  For $j = 1$, we also have that $\sigma_1 = \rho_1$, and this sequence consists of one price with value $u_1$ and one price with value $L$, followed by $u_{2} - \epsilon$ repeated $k$ times and $L$ repeated $k$ times.

\smallskip
Observe that as $\epsilon \rightarrow 0$, $\sigma_j$ and $\rho_j$ are sequences yielding the worst-case ratios in $\mathcal{S}_j$, as \DTPRmax~is forced to accept $(k-j)$ worst-case $L$ values at the end of the sequence, and each accepted value is exactly equal to the corresponding threshold.

$\sigma_j$ and $\rho_j$ also represent two extreme possibilities for the switching cost.  In $\sigma_j$, \DTPRmax~only switches twice, but it mostly accepts values $\ell_i$.  In $\rho_j$, \DTPRmax~must switch $j + 1$ times because there are many intermediate $L$ values, but it only accepts values which are at least $u_i$.  

Observe that $\texttt{OPT}(\sigma_j)/\DTPRmax(\sigma_j) = \texttt{OPT}(\rho_j)/\DTPRmax(\rho_j)$.  First, the optimal solution for both sequences is exactly the same: $k c_{\max}(\sigma_j) - 2\beta = k c_{\max}(\rho_j) - 2\beta$.\\
For any sequence $s$ in $\mathcal{S}_j$, we also know that $c_{\max}(s) < u_{j+1}$, so $\texttt{OPT}(\rho_j) = \texttt{OPT}(\sigma_j) \leq k u_{j+1} - 2\beta$.

By definition of the threshold families $\{u_i\}_{i \in [1, k]}$ and $\{\ell_i\}_{i \in [1, k]}$, we know that\\ $\sum_{i=1}^j~u_i~-~j2\beta~=~\sum_{i=1}^j \ell_i$ for any value $j \geq 2$:
\begin{align*}
\DTPRmax(\rho_j) = \left( u_1 + \sum_{i=2}^j \ell_i + (k-j)L - 4\beta \right) = \left( \sum_{i=1}^j u_i + (k-j)L - (j+1)2\beta \right) = \DTPRmax(\sigma_j).
\end{align*}
Note that whenever $j < 2$, we have that $\sigma_0 = \rho_0$, and $\sigma_1 = \rho_1$.  Thus, $\DTPRmin(\rho_j) = \DTPRmin(\sigma_j)$ holds for any value of $j$.

By definition of $u_1$, we simplify $u_1 + \sum_{i=2}^j \ell_i + (k-j)L - 4\beta$ to $\sum_{i=1}^j \ell_i + (k-j)L - 2\beta$.  For any sequence $s \in \mathcal{S}_j$, we have the following:
\begin{align}
\frac{\texttt{OPT}(s)}{\DTPRmax(s)} \leq \frac{\texttt{OPT}(\sigma_j)}{\DTPRmax(\sigma_j)} = \frac{\texttt{OPT}(\rho_j)}{\DTPRmax(\rho_j)} \leq \frac{ku_{j+1} - 2\beta}{\sum_{i=1}^j \ell_i + (k-j)L - 2\beta}.
\end{align}
\begin{lemma} \label{lem:intermedStepMax}
For any $j \in [0, k]$, by definition of $\{u_i\}_{i \in [1, k]}$ and $\{\ell_i\}_{i \in [1, k]}$,
\begin{align*}
\omega \cdot \left( \sum_{i=1}^j \ell_i + (k-j)L - 2\beta \right) \leq k u_{j+1} - 2\beta. \quad \quad \textit{The proof is deferred to Appendix \ref{appendix:proofs}.}
\end{align*}
\end{lemma}
\noindent For $\epsilon \rightarrow 0$, the competitive ratio $\texttt{OPT}/\DTPRmax$ is exactly $\omega$:
\begin{align*}
\forall 0 \leq j \leq k: \;\;\; \frac{\texttt{OPT}(\sigma_j)}{\DTPRmax(\sigma_j)} = \frac{ku_{j+1} - 2\beta}{\sum_{i=1}^j \ell_i + (k-j)L - 2\beta} = \omega.
\end{align*}
and thus for any sequence $s \in \mathcal{S}$,
\begin{align*}
\forall s \in \mathcal{S}: \;\;\; \frac{k c_{\max}(s) - 2\beta}{\DTPRmax(s)} \leq \omega.
\end{align*}
Since $\texttt{OPT}(s) \leq k c_{\max}(s) - 2\beta$ for any sequence $s$, this implies that $\DTPRmax$ is $\omega$-competitive.\end{proof}

\begin{proof}[Proof of Corollary~\ref{cor:max}]
For simplification purposes, let $\beta = bL/2$, where $b$ is a real constant on the interval $\left(0, k \right)$.  To show part \textbf{(a)} for \texttt{REGIME-1}, with fixed $k \geq 1$, observe that for sufficiently large $\omega$, we have the following:
\smallskip
\begin{align*}
\theta - b - 1 = \left( \omega - 1 \right) \left( 1 + \frac{\omega}{k} \right)^k - \left(b - \frac{b}{k} + \frac{b\omega}{k} \right) \left( 1 + \frac{\omega}{k} \right)^k \approx (1 + o(1)) \left[ \omega \left( \frac{\omega}{k} \right)^k - b \left( \frac{\omega}{k} \right)^{k+1} - b \right].
\end{align*}
\smallskip
Let $\omega_{+} = \sqrt[k+1]{k^k \cdot \frac{k\theta}{k - b}}$.  \;\; Then, for sufficiently large $\omega$, we have the following:
\smallskip
\begin{align*}
(1 + o(1)) \left[ \omega_{+} \left( \frac{\omega_{+}}{k} \right)^k - b \left( \frac{\omega_{+}}{k} \right)^{k+1} - b \right] = (1 + o(1)) \frac{(k - b) (\theta)}{k - b} = (1 + o(1)) \left[\theta - b \right].
\end{align*}
\smallskip
Furthermore, let $\varepsilon > 0$ and set $\omega_{-} = (1 - \varepsilon) \sqrt[k+1]{k^k \cdot \frac{k\theta}{k - b}}$.\\  A similar calculation as above shows that for sufficiently large $\theta$ we have:
\begin{align*}
\left( \omega_{-} - 1 \right) \left( 1 + \frac{\omega_{-}}{k} \right)^k - \left(b - \frac{b}{k} + \frac{a\omega_{-}}{k} \right) \left( 1 + \frac{\omega_{-}}{k} \right)^k \geq ( 1 - 3k\varepsilon) \left[\theta - b \right].
\end{align*}
\smallskip
Thus, $\omega = O \left( \sqrt[k+1]{k^k \frac{k\theta}{k - b}} \right)$ satisfies (\ref{eq:omega}) for sufficiently large $\omega$, fixed $k \geq 1$, and\\ $\beta = \frac{bL}{2} \text{ s.t. } b \in (1, k)$.  
\smallskip

To show part \textbf{(b)} for \texttt{REGIME-2}, observe that the right-hand side of (\ref{eq:omega}) can be approximated as $\left( 1 + \frac{\omega}{k} \right)^k~\approx~e^{\omega}$ when $k \rightarrow \infty$.  Then by taking limits on both sides, we obtain the following:
\begin{align*}
\frac{U - L - 2\beta}{L \left( \omega - 1 \right) - 2\beta\left(1 \right)} = e^{\omega}.
\end{align*}
Let $\beta = bL/2$ as outlined above.  We then obtain the following:
\begin{align*}
\frac{U - L - bL}{L \left( \omega - 1 \right) - bL} = \frac{\theta - 1 - b}{ \omega - 1 - b} = e^{\omega} \; \Longrightarrow \; \theta - 1 - b = \left( \omega - 1 - b\right) e^{\omega}.
\end{align*}
By definition of the Lambert $W$ function, solving this equation for $\omega$ obtains part (2).
\end{proof}

\section{Proofs of Lemmas~\ref{lem:intermedStepMin} and~\ref{lem:intermedStepMax}} \label{appendix:proofs}
In this section, we give the deferred proofs of Lemmas~\ref{lem:intermedStepMin} and~\ref{lem:intermedStepMax}, which are used in the proofs of Theorem~\ref{thm:compPRmin} and Theorem~\ref{thm:compPRmax}, respectively.

\begin{proof}[Proof of Lemma~\ref{lem:intermedStepMin}]
We show that the following holds for any $j \in [0, k]$, by Definition \ref{def:minthres}:
\begin{align*}
\sum_{i=1}^j u_i + (k-j)U + 2\beta &\leq \alpha \cdot (k\ell_{j+1} + 2\beta).
\end{align*}
\noindent First, note that $k\ell_{j+1} = k (u_{j+1} - 2\beta)$ for all $j \in [0, k]$, by Observation \ref{obs:thresholdDifference}.  This gives us the following:
\begin{align*}
\sum_{i=1}^j u_i + (k-j)U + 2\beta &\leq \alpha ku_{j+1} + \alpha 2\beta - \alpha k2\beta,\\
\sum_{i=1}^j u_i + (k-j)U + \left[ 2\beta - \alpha 2\beta + \alpha k2\beta \right] &\leq \alpha ku_{j+1},\\
\frac{(k-j)U}{\alpha k} + \frac{\sum_{i=1}^j u_i}{\alpha k} + \left[ \frac{2\beta}{\alpha k} - \frac{2\beta}{k} + 2\beta \right] &\leq u_{j+1}.
\end{align*}
\noindent By substituting Def. \ref{def:minthres} into $\sum_{i=1}^j u_i$, the above can be simplified exactly to the closed form for $u_{j+1}$:
\begin{align*}
\frac{U}{\alpha} - \frac{jU}{\alpha k} + \left( \frac{\sum_{i=1}^j u_i}{\alpha k} \right) &+ \left[ \frac{2\beta}{\alpha k} - \frac{2\beta}{k} + 2\beta \right] = u_{j+1},\\
\left[ U - \left( U - \frac{1}{\alpha} \right) \left( 1 + \frac{1}{\alpha k} \right)^{j} \right] &+ \left[ \left( \frac{2\beta}{\alpha k} - \frac{2\beta}{k} + 2\beta \right) \left( 1 + \frac{1}{\alpha k} \right)^{j} \right] = u_{j+1}.
\end{align*}
and the claim follows by the definition of $u_{j+1}$.
\end{proof}

\begin{proof}[Proof of Lemma~\ref{lem:intermedStepMax}]
We show that the following holds for any $j \in [0, k]$, by Definition \ref{def:maxthres}:
\begin{align*}
\omega \cdot \left( \sum_{i=1}^j \ell_i + (k-j)L - 2\beta \right) \leq k u_{j+1} - 2\beta.
\end{align*}
\noindent First, note that $k u_{j+1} = k (\ell_{j+1} + 2\beta)$ for all $j \in [0, k]$, by Observation \ref{obs:thresholdDifference}.  This gives us the following:
\begin{align*}
\sum_{i=1}^j \ell_i + (k-j)L - 2\beta &\leq \frac{k \ell_{j+1}}{\omega} - \frac{2\beta}{\omega} + \frac{k2\beta}{ \omega },\\
\sum_{i=1}^j \ell_i + (k-j)L - \left[ 2\beta - \frac{2\beta}{\omega} + \frac{k2\beta}{ \omega } \right] &\leq \frac{k \ell_{j+1}}{\omega},\\
\frac{\omega \left( \sum_{i=1}^j \ell_i \right) }{k} + \frac{\omega (k-j)L}{k} - \left[ \frac{\omega 2\beta}{k} - \frac{2\beta}{k} + 2\beta \right] &\leq \ell_{j+1}.\\ 
\end{align*}
\noindent By substituting Def. \ref{def:maxthres} into $\sum_{i=1}^j \ell_i$, the above can be simplified exactly to the closed form for $\ell_{j+1}$:
\begin{align*}
\omega L - \frac{\omega jL}{k} + \frac{\omega \left( \sum_{i=1}^j \ell_i \right) }{k} - \left[ \frac{\omega 2\beta}{k} - \frac{2\beta}{k} + 2\beta \right] &= \ell_{j+1},\\ 
\left[ L + (\omega L - L) \left(1 + \frac{\omega}{k} \right)^{j} \right] - \left[ \left( \frac{\omega 2\beta}{k} - \frac{2\beta}{k} + 2\beta \right) \left(1 + \frac{\omega}{k} \right)^{j} \right] &= \ell_{j+1}.
\end{align*}
and the claim follows by the definition of $\ell_{j+1}$.
\end{proof}

\section{Proofs of Lower Bound Results}
\label{app:lowerbound}
This section formally proves the lower bound results for both \OPRmin and \OPRmax, building on the proof sketch provided in Section \ref{sec:prooflowerbound}.

\subsection{Proof of Theorem \ref{thm:lowerboundmin} (\OPRmin Lower Bound)}
\label{app:lowerboundmin}

\begin{proof}[Proof of Theorem \ref{thm:lowerboundmin}]
Let $\texttt{ALG}$ be a deterministic online algorithm for \OPRmin, and suppose that the adversary uses the price sequence $\ell_1, \dots, \ell_k$, which is exactly the sequence defined by~(\ref{eq:minthres}).  $\ell_1$ is presented to $\texttt{ALG}$, at most $k$ times or until $\texttt{ALG}$ accepts it.  If $\texttt{ALG}$ never accepts $\ell_1$, the remainder of the sequence is all $U$, and $\texttt{ALG}$ achieves a competitive ratio of $\frac{kU + 2\beta}{k\ell_1 + 2\beta} = \alpha$, as defined in~\eqref{eq:balancemin}.

If $\texttt{ALG}$ accepts $\ell_1$, the next price presented is $U$, repeated at most $k$ times \textit{or until $\texttt{ALG}$ switches to reject $U$}.  After $\texttt{ALG}$ has switched, $\ell_2$ is presented to $\texttt{ALG}$, at most $k$ times or until $\texttt{ALG}$ accepts it.  Again, if $\texttt{ALG}$ never accepts $\ell_2$, the remainder of the sequence is all $U$, and $\texttt{ALG}$ achieves a competitive ratio of at least $\frac{\ell_1 + (k-1)U + 4\beta}{k\ell_2 + 2\beta} = \alpha$, as defined in~\eqref{eq:balancemin}.

As the sequence continues, whenever $\texttt{ALG}$ does not accept some $\ell_i$ after it is presented $k$ times, the adversary increases the price to $U$ for the remainder of the sequence.  Otherwise, if $\texttt{ALG}$ accepts $k$ prices before the end of the sequence, the adversary concludes by presenting $L$ at least $k$ times.

Observe that any $\texttt{ALG}$ which does not immediately reject the first $U$ presented to it after accepting some $\ell_i$ obtains a competitive ratio strictly worse than $\alpha$.  To illustrate this, suppose $\texttt{ALG}$ has just accepted $\ell_1$, incurring a cost of $\ell_1 + \beta$ so far.  The adversary begins to present $U$, and $\texttt{ALG}$ accepts $y \leq (k-1)$ of these $U$ prices before switching away.  If $y = (k-1)$, $\texttt{ALG}$ will accept $k$ prices before the end of the sequence and achieve a competitive ratio of $\frac{\ell_1 + (k-1)U + 2\beta}{kL + 2\beta} > \alpha$.  Otherwise, if $y < (k-1)$, the cost incurred by $\texttt{ALG}$ so far is at least $\ell_1 + 2\beta + yU$, while the cost incurred by $\texttt{ALG}$ if it had immediately switched away ($y = 0$) would be $\ell_1 + 2\beta$ -- since any price which might be accepted by $\texttt{ALG}$ in the future should be $\leq U$, the latter case strictly improves the competitive ratio of $\texttt{ALG}$.  

Assuming that $\texttt{ALG}$ does immediately reject any $U$ presented to it, and that $\texttt{ALG}$ accepts some prices before the end of the sequence, the competitive ratio attained by $\texttt{ALG}$ is at least $\frac{\sum_{i=1}^j \ell_i + (j+1)2\beta + (k-j)U}{k\ell_{j+1} + 2\beta} = \alpha$, as defined in~\eqref{eq:balancemin}.  

Similarly, if $\texttt{ALG}$ accepts $k$ prices before the end of the sequence, the competitive ratio attained by $\texttt{ALG}$ is at least $\frac{\sum_{i=1}^k \ell_i + k2\beta}{kL + 2\beta} = \alpha$, as defined in~\eqref{eq:balancemin}.

Since any arbitrary deterministic online algorithm $\texttt{ALG}$ cannot achieve a competitive ratio better than $\alpha$ playing against this adaptive adversary, our proposed algorithm \DTPRmin~is optimal.
\end{proof}

\subsection{Proof of Theorem \ref{thm:lowerboundmax} (\OPRmax Lower Bound)}
\label{app:lowerboundmax}
\begin{proof}[Proof of Theorem \ref{thm:lowerboundmax}]
Let $\texttt{ALG}$ be a deterministic online algorithm for \OPRmax, and suppose that the adversary uses the price sequence $u_1, \dots, u_k$, which is exactly the sequence defined by~(\ref{eq:maxthres}).  $u_1$ is presented to $\texttt{ALG}$, at most $k$ times or until $\texttt{ALG}$ accepts it.  If $\texttt{ALG}$ never accepts $u_1$, the remainder of the sequence is all $L$, and $\texttt{ALG}$ achieves a competitive ratio of $\frac{ku_1 - 2\beta}{kL - 2\beta} = \omega$, as defined in~\eqref{eq:balancemax}.

If $\texttt{ALG}$ accepts $u_1$, the next price presented is $L$, repeated at most $k$ times \textit{or until $\texttt{ALG}$ switches to reject $L$}.  After $\texttt{ALG}$ has switched, $u_2$ is presented to $\texttt{ALG}$, at most $k$ times or until $\texttt{ALG}$ accepts it.  Again, if $\texttt{ALG}$ never accepts $u_2$, the remainder of the sequence is all $L$, and $\texttt{ALG}$ achieves a competitive ratio of at least $\frac{ku_2 - 2\beta}{u_1 + (k-1)L - 4\beta} = \omega$, as defined in~\eqref{eq:balancemax}.

As the sequence continues, whenever $\texttt{ALG}$ does not accept some $u_i$ after it is presented $k$ times, the adversary drops the price to $L$ for the remainder of the sequence.  Otherwise, if $\texttt{ALG}$ accepts $k$ prices before the end of the sequence, the adversary concludes by presenting $U$ at least $k$ times.

Observe that any $\texttt{ALG}$ which does not immediately reject the first $L$ presented to it after accepting some $u_i$ obtains a competitive ratio strictly worse than $\omega$.  To illustrate this, suppose $\texttt{ALG}$ has just accepted $u_1$, achieving a profit of $u_1 - \beta$ so far.  The adversary begins to present $L$ prices, and $\texttt{ALG}$ accepts $y \leq (k-1)$ of these $L$ prices before switching away.  If $y = (k-1)$, $\texttt{ALG}$ will accept $k$ prices before the end of the sequence and achieve a competitive ratio of $\frac{kU - 2\beta}{u_1 + (k-1)L - 2\beta} > \omega$.  Otherwise, if $y < (k-1)$, the profit achieved by $\texttt{ALG}$ so far is at most $u_1 - 2\beta + yL$, while the profit achieved by $\texttt{ALG}$ if it had immediately switched away ($y = 0$) would be $u_1 - 2\beta$ -- since any price which might be accepted by $\texttt{ALG}$ in the future should be $\geq L$, the latter case strictly improves the competitive ratio of $\texttt{ALG}$.  

Assuming that $\texttt{ALG}$ does immediately reject any $L$ presented to it, and that $\texttt{ALG}$ accepts some prices before the end of the sequence, the competitive ratio attained by $\texttt{ALG}$ is at least $\frac{ku_{j+1} - 2\beta}{\sum_{i=1}^j u_i - (j+1)2\beta + (k-j)L} = \omega$, as defined in (\ref{eq:balancemax}).   

Similarly, if $\texttt{ALG}$ accepts $k$ prices before the end of the sequence, the competitive ratio attained by $\texttt{ALG}$ is at least $\frac{kU - 2\beta}{\sum_{i=1}^k u_i - k2\beta} = \omega$, as defined in (\ref{eq:balancemax}).

Since any arbitrary deterministic online algorithm $\texttt{ALG}$ cannot achieve a competitive ratio better than $\omega$ playing against this adaptive adversary, our proposed algorithm \DTPRmax~is optimal.
\end{proof}